\documentclass[11 pt, letterpaper]{article}

\usepackage{amsmath}
\usepackage{amssymb}
\usepackage{amsthm}
\usepackage{mathtools}
\usepackage{tikz}
\usetikzlibrary{positioning}
\usepackage{fancybox}
\usepackage{multirow}
\usepackage{multicol}
\usepackage{caption}
\usepackage{relsize}
\usepackage{color, graphicx}

\usepackage{algorithm}
\usepackage{algpseudocode}

\usepackage[margin = 1 in]{geometry}
\usepackage{url}
\usepackage{hyperref}

\theoremstyle{plain}
\newtheorem{theorem}{Theorem}
\newtheorem{definition}{Definition}
\newtheorem{lemma}{Lemma}
\newtheorem{corollary}{Corollary}

\theoremstyle{remark}
\newtheorem{remark}{Remark}

\newcommand{\defeq}{\stackrel{\mathsmaller{\mathsf{def}}}{=}}
\newcommand{\E}{\text{E}}

\def\poly{\operatorname{poly}}

\def\Reviewer: #1\par{\clearpage\section*{Reviewer #1}}

\def\r#1\par{\begin{quote} \textcolor{BrickRed}{\texttt{#1}} \end{quote}}

\newboolean{short}
\setboolean{short}{false}

\newcommand{\onlyShort}[1]{\ifthenelse{\boolean{short}}{#1}{}}
\newcommand{\onlyLong}[1]{\ifthenelse{\boolean{short}}{}{#1}}

\title{Sleeping is Efficient: MIS in $O(1)$-rounds Node-averaged Awake Complexity}

\author{Soumyottam Chatterjee \thanks{Georgetown University, Washington, D.C. Email: \href{mailto: sc1943@georgetown.edu}{\texttt{sc1943@georgetown.edu}}. Research supported in part by NSF award CCF-1733842} \and Robert Gmyr \thanks{Microsoft, Redmond, Washington. Email: \href{mailto: robert.gmyr@microsoft.com}{\texttt{robert.gmyr@microsoft.com}}. Work done while at University of Houston.} \and Gopal Pandurangan \thanks{University of Houston, Houston, Texas. Email: \href{mailto: gopal@cs.uh.edu}{\texttt{gopal@cs.uh.edu}}. Research supported in part by NSF grants IIS-1633720, CCF-1540512, and CCF-1717075, and by BSF grant 2016419.}}

\begin{document}

\maketitle

\thispagestyle{empty}

\begin{abstract}

Maximal Independent Set (MIS) is one of the fundamental problems in distributed computing. The round (time) complexity of distributed MIS has traditionally focused on the \emph{worst-case 
time} for all nodes to finish. The best-known (randomized) MIS algorithms take $O(\log{n})$ worst-case rounds on general graphs (where $n$ is the number of nodes). Breaking the $O(\log{n})$ worst-case bound has been a longstanding open problem, while currently the best-known lower bound is $\Omega(\sqrt{\frac{\log{n}}{\log{\log{n}}}})$ rounds.

Motivated by the goal to reduce \emph{total} energy consumption in energy-constrained networks such as sensor and ad hoc wireless networks, we take an alternative approach to measuring performance. We focus on minimizing the total (or equivalently, the \emph{average}) time for all nodes to finish. It is not clear whether the currently best-known  algorithms yield constant-round (or even $o(\log{n})$) node-averaged round complexity for MIS in general graphs. We posit the \emph{sleeping model}, a generalization of the traditional model, that allows nodes to enter either ``sleep'' or ``waking'' states at any round. While waking state corresponds to the default state in the traditional model, in sleeping state a node  is ``offline'', i.e., it does not send or receive messages (and messages sent to it are dropped as well) and does not incur any time, communication, or local computation cost.  Hence, in this model, only rounds in which a node is awake are counted and we are interested in minimizing the average as well as the worst-case number of rounds a node spends in the awake state, besides the traditional worst-case round complexity (i.e., the rounds for all nodes to finish including both the awake and sleeping rounds).

Our main result is that we show that {\em MIS can be solved in (expected) $O(1)$ rounds under node-averaged awake complexity measure} in the sleeping model. In particular, we present a randomized distributed  algorithm for MIS that has expected {\em $O(1)$-rounds node-averaged awake complexity} and, with high probability has {\em $O(\log{n})$-rounds worst-case awake complexity} and {\em $O(\log^{3.41}n)$-rounds worst-case complexity}.

Our work is a step towards understanding the node-averaged complexity of MIS both in the traditional and sleeping models. It is also the first step in designing energy-efficient distributed algorithms for energy-constrained networks.\\


\end{abstract}

\newpage

\setcounter{page}{1}


\section{Introduction} \label{sec:intro}

Computing a \emph{maximal independent set} or \emph{MIS} problem in a network is one of the central problems in distributed computing. About 35 years ago, Alon, Babai, and Itai \cite{Alon_1986} and Luby \cite{Luby_1986} presented a randomized distributed algorithm for MIS, running on $n$-node graphs in $O(\log{n})$ rounds with high probability.\footnote{Throughout, we use ``with high probability (whp)'' to mean with probability at least $1 - n^{-\gamma}$, for some constant $\gamma > 1$.} Since then the MIS problem has been studied extensively, and recently there has been some exciting progress in designing faster distributed MIS algorithms. For $n$-node graphs with maximum degree $\Delta$, Ghaffari \cite{Ghaffari_2016_SODA} presented a randomized MIS algorithm running in \onlyShort{$O(\log{\Delta}) + 2^{O(\sqrt{\log{\log{n}}})}$ rounds,}
\onlyLong
{
    \begin{center}
        $O(\log \Delta) + 2^{O(\sqrt{\log\log n})}$ rounds,
    \end{center}    
}
improving over the algorithm of Barenboim et al.\ \cite{Barenboim_2016} that runs in \onlyShort{$O(\log^2 \Delta) + 2^{O(\sqrt{\log\log n})}$ rounds.}
\onlyLong
{
    \begin{center}
        $O(\log^2 \Delta) + 2^{O(\sqrt{\log\log n})}$ rounds.
    \end{center}
}
It was further improved by Rozhon and Ghaffari to \onlyShort{$O(\log{\Delta}  +  \poly{(\log{\log{n}})})$}
\onlyLong
{
    \begin{center}
        $O(\log{\Delta}  +  \poly{(\log{\log{n}})})$
    \end{center}
}
rounds \cite[Corollary $3.2$]{Rozhon_2020}.

While the above results constitute a significant improvement in our understanding of the round complexity of the MIS problem, it should be noted that in general graphs, the best-known running time is still $O(\log{n})$ (even for randomized algorithms). Furthermore, there is a lower bound of \onlyShort{$\Omega(\min\{\frac{\log \Delta}{\log \log \Delta}, \sqrt{\frac{\log n}{\log \log n}}\})$ due to Kuhn et al.\ \cite{Kuhn_2016} that also applies to randomized algorithms. Thus, for example, say, when $\Delta = \Theta(2^{\sqrt{\log{n}}})$, it follows that one cannot hope for algorithms faster than  $\sqrt{\frac{\log n}{\log \log n}}$ rounds. Balliu et al.\ showed recently that one cannot hope for algorithms that run within $o(\Delta) + O(\log^*{n})$ rounds for the regimes where $\Delta << \log \log n$ (for randomized algorithms) \cite[Corollary $5$]{Balliu_2019} and $\Delta << \log n$ (for deterministic algorithms) \cite[Corollary $6$]{Balliu_2019}.
}

\onlyLong
{
    \begin{center}
        $\Omega(\min\{\frac{\log \Delta}{\log \log \Delta}, \sqrt{\frac{\log n}{\log \log n}}\})$
    \end{center}

    due to Kuhn et al.\ \cite{Kuhn_2016} that also applies to randomized algorithms. Thus, for example, say, when $\Delta = \omega(2^{\sqrt{\log n}})$, it follows that one cannot hope for algorithms faster than  $\sqrt{\frac{\log n}{\log \log n}}$ rounds. Balliu et al.\ showed recently that one cannot hope for algorithms that run within $o(\Delta) + O(\log^*n)$ rounds for the regimes where $\Delta << \log \log n$ (for randomized algorithms) \cite[Corollary $5$]{Balliu_2019} and $\Delta << \log n$ (for deterministic algorithms) \cite[Corollary $6$]{Balliu_2019}.
}
\subsection{Energy considerations, Sleeping model, and Node-averaged round complexity} \label{sec:sleeping}

It is important to note that all prior works on MIS, including the ones mentioned above, are focused on measuring the \emph{worst-case} number of rounds for nodes to finish. In other words, the time complexity is measured as the time (number of rounds) needed for the last (slowest) node(s) to finish. As mentioned above, the best-known bound for this measure is still $O(\log{n})$ for general graphs (even for randomized algorithms). In this paper, we take an alternative approach to designing MIS algorithms motivated by two main considerations.

The first consideration is the motivation of designing {\em energy-efficient} algorithms for ad hoc wireless and sensor networks. In such networks, a node's energy consumption depends on the amount of time it is actively communicating with nodes; more importantly, significant energy is spent by a node even when it is just {\em idle}, i.e., waiting to hear from a neighbor. Experimental results show that the energy consumption  in an idle state is only slightly smaller than that in a transmitting or receiving state \cite{Zheng_2005, Feeney_2001}. Thus, even though there might be no messages exchanged between a node and its sender, a node might be spending quite a bit of energy if it is just waiting to receive a message.

On the other hand,  the energy consumption in the ``sleeping'' state, i.e., when it has switched off its communication devices and is not sending, receiving or listening, is significantly less than in the transmitting/receiving/idle (listening) state (see e.g., \cite{Zheng_2005, Feeney_2001, King_2011, Wang_2006, Yang_2013}). A node may cleverly enter and exit sleeping mode to save energy during the course of an algorithm. In fact, this has been exploited by  protocols to save power in ad hoc wireless networks by judiciously switching between two states --- \emph{sleeping} and \emph{awake} --- as needed (the MAC layer provides support for switching between states \cite{Zheng_2005, Yang_2013, Murthy_2004_Book}).

The second consideration, motivated by the first, is saving the {\em total} amount of energy spent by the nodes during the course of an algorithm. Note that in sleeping mode, we assume that there is no energy spent. Thus the total energy is measured as proportional to the {\em total time (number of rounds)} that nodes have spent in the ``awake'' or ``normal'' mode (i.e., non-sleeping mode). In this paper, we thus focus on minimizing the total number of rounds --- or equivalently the {\em average} number of rounds --- spent by all nodes in their {\em awake state} during an algorithm. Our goal is to design distributed algorithms with low {\em node-averaged awake complexity} (see Section \ref{sec:model}).

Motivated by the above considerations, we posit the {\em sleeping model} for distributed algorithms which is a generalization of the traditional model (a more detailed description is given in Section \ref{sec:model}). In the sleeping model, a node can be in either of the two states --- {\em sleeping} or {\em awake (or normal)}. While in the traditional model nodes are only in the awake state, in the sleeping model nodes have the {\em option} of entering sleeping state at any round as well as exiting the sleeping state and entering the awake state at a later round. In the sleeping state, a node does not send or receive messages and messages sent to it by other nodes are lost; it also does not do any local computation. If a node enters a sleeping state, then it is assumed that it does not incur any time or message cost (or other resource costs, such as energy). Some previous models (see e.g., \cite{Khan_2009} and the references therein) assumed that nodes incur little or no energy only when they are not sending/receiving messages; however, this is not true in real-world ad hoc wireless and sensor networks, where considerable energy is spent even when nodes are ``idle'' or ``listening'' for messages. The sleeping model is more realistic, since in the sleeping state nodes turn off their communication (e.g., wireless) devices fully. However, it becomes more challenging to design efficient algorithms under this model.
\subsection{Model and Complexity Measures} \label{sec:model}

Before we define the sleeping model, we will recall the traditional model used in distributed algorithms.


\paragraph{Traditional Model.}

We consider the standard synchronous \textsf{Congest} model \cite{Peleg_2000_Book}, where nodes are always ``awake'' from the start of the algorithm (i.e., round zero). We are given a distributed network of $n$ nodes, modeled as an undirected graph $G$. Each node hosts a processor with limited initial knowledge. We assume that nodes have unique \texttt{ID}s,\footnote{Making this assumption is not essential, but it simplifies presentation.} and at the beginning of the computation each node is provided its \texttt{ID} as input. We assume that each node has ports (each port having a unique port number); each incident edge is connected to one distinct port. We also assume that nodes know $n$, the number of nodes in the network. Thus, a node has only {\em local} knowledge.

Nodes are allowed to communicate through the edges of the graph $G$ and it is assumed that communication is synchronous and occurs in rounds. In particular, we assume that each node knows the current round number (starting from round 0). In each round, each node can perform some local computation (which finishes in the same round) including accessing a private source of randomness, and can exchange (possibly distinct) $O(\log{n})$-bit messages with each of its neighboring nodes.

This model of distributed computation is called the \textsf{Congest}$(\log{n})$ model or simply the \textsf{Congest} model \cite{Peleg_2000_Book}. We note that our algorithms also, obviously apply to the \textsf{Local} model, another well-studied model \cite{Peleg_2000_Book} where there is no restriction on the size of the messages sent per edge per round. \onlyLong{ The \textsf{Local} (resp.\ \textsf{Congest}) model does not put any constraint on the computational power of the nodes, but we do not abuse this aspect: our algorithms perform only light-weight computations.}

Some distributed MIS algorithms (e.g., \cite{Panconesi_1992, Ghaffari_2016_SODA, Barenboim_2016}) assume the \textsf{Local} model, whereas others work also in the \textsf{Congest} model \cite{Luby_1986, Pai_2017, Ghaffari_2019_SODA}.


\paragraph{Sleeping Model.}

We augment the traditional \textsf{Congest} (or \textsf{Local}) model by allowing nodes to enter a {\em sleeping state} at any round. In the sleeping model, a node can be in either of the two states before it finishes executing the algorithm (locally) --- in other words, before it enters a final ``termination'' state. That is, any node $v$, can decide to {\em sleep} starting at any (specified) round of its choice; we assume all nodes know the correct round number whenever they are awake.\footnote{\label{ft:adhoc} One way to implement the sleeping model is to assume that a node's local (synchronized) clock is always running; a node before it enters the sleeping state, sets an ``interrupt'' (alarm) to wake at a specified later round. In practice, the IEEE 802.11 MAC provides low-level support for power management and synchronizing nodes to wake up for data delivery \cite{Zheng_2005, Murthy_2004_Book, Yang_2013}.} It can {\em wake up} again later at any specified  round --- this is the {\em awake} state. In the sleeping state, a node can be considered ``dead'' so to speak: it does not send or receive messages, nor it does any local computation. Messages sent to it by other nodes when it was sleeping are lost. However, a node can awake itself at any specified later round. Note that in the traditional model, which can be considered as a special case of the sleeping model, nodes are always in the awake state. In the sleeping model, a node can potentially conserve its resources by judiciously determining if, when, and how long to sleep.


\paragraph{Node-averaged Round Complexity.}

For a distributed algorithm  on a network $G = (V, E)$ in the sleeping model, we are primarily interested in the ``node-averaged awake round complexity'' or simply ``node-averaged awake complexity''. For a deterministic algorithm, for a node $v$, let $a_v$ be the number of ``awake'' rounds needed for $v$ to finish, i.e., $a_v$ {\em only counts the number of rounds in the awake state} of $v$. Then we define the {\em node-averaged awake complexity} to be $\frac{1}{n}\sum_{v \in V} a_v$.

For a randomized algorithm, for a node $v$, let $A_v$ be the random variable denoting the number of awake rounds needed for $v$ to finish. Then let the random variable $A$ be defined as $\frac{1}{n}\sum_{v \in V} A_v$, i.e., the average of the $A_v$ random variables. Then the expected {\em node-averaged awake complexity} of the randomized algorithm is
\begin{center}
    $\E[A] =  \E[\frac{1}{n}\sum_{v \in V} A_v] = \frac{1}{n}\sum_{v \in V} \E[A_v]$.
\end{center}
In this paper, we are mainly focused on this measure. However, one can also study other properties of $A$, e.g., high probability bounds on $A$.

Note that analogous definitions also naturally apply to the {\em node-averaged round complexity}\footnote{Note that henceforth when we don't specify ``awake'' in the complexity measure, it means that we are referring to the traditional model, and if we do, we are referring to the sleeping model.} in the {\em traditional model}, where all rounds are counted (since nodes are always awake).


\paragraph{Worst-case Round Complexity.}

We measure the ``worst-case awake round complexity'' (or simply the ``worst-case awake complexity'') in the sleeping model as the {\em worst-case} number of {\em awake} rounds (from the start) taken by a node to finish the algorithm. That is, if $a_v$ be the number of awake rounds of $v$ before it terminates, then the {\em worst-case awake complexity} is $\max_{v \in V} a_v$.

While our goal is to design distributed algorithms that are efficient with respect to \emph{node-averaged awake complexity}, we would also like them to be efficient (as much as possible) with respect to the \emph{worst-case awake complexity}, as well as the traditional {\em worst-case round complexity}, where {\em all rounds} (including rounds spent in sleeping state) are  counted.
\subsection{MIS in $O(1)$-rounds node average complexity?} \label{sec:focus}

In light of the difficulty in breaking the $o(\log{n})$-round (traditional) worst-case barrier and the $\Omega(\min\{\frac{\log \Delta}{\log \log \Delta}, \frac{\sqrt{\log n}}{\log \log n}\})$ lower bound for worst-case round complexity, as well as motivated by energy considerations discussed above, a fundamental question that we seek to answer is this: 
\begin{center}
    \emph{Can we design a distributed MIS algorithm that takes $O(1)$-rounds node-averaged awake complexity?} 
\end{center}

Before we answer this question, it is worth studying the node-averaged round complexity of some well-known distributed MIS algorithms in the \emph{traditional model}. It is not clear whether Luby's algorithms (both versions of it \cite{dnabook, Luby_1986}) give $O(1)$-round node-averaged complexity, or even $o(\log{n})$-round node-averaged complexity. The same is the situation with the algorithms of Alon et al.\ \cite{Alon_1986} and Karp et al.\ \cite{Karp_1985} as well as known deterministic MIS algorithms \cite{Awerbuch_1989, Panconesi_1992} (see also \cite{Barenboim_2016, Rozhon_2020}). The algorithm of \cite{Ghaffari_2016_SODA} (also of Barenboim et al.\ \cite{Barenboim_2016}) does not seem to give $O(1)$ (or even $o(\log{n})$) node-averaged complexity. For example, take Ghaffari's algorithm \cite{Ghaffari_2016_SODA}, which is well-suited for analyzing the node-averaged complexity since it is ``node centric'': for any node $v$, it gives a probabilistic bound on when $v$ will finish. More precisely, it shows that for each node $v$, the probability that $v$ has  not finished (i.e., its status has not been determined) after $O(\log{(deg(v))} + \log{(\frac{1}{\epsilon})})$ rounds is at most $\epsilon$. Using this it is easy to compute the (expected) node-averaged complexity of Ghaffari's algorithm. However, this is still only $O(\log{n})$, as $\log{(deg(v))}$ can be $\Theta(\log{n})$ for most nodes.

Recently Barenboim and Tzur \cite{Barenboim_2019} showed that MIS can be solved in $O(a + \log^*{n})$ rounds under {\em node-averaged complexity} deterministically, where $a$ is the arboricity of the graph. It is a open question whether one can design an algorithm with $O(1)$ (or even $o(\log{n})$) node-averaged round complexity in the \emph{traditional model} for general graphs (which can have arboricity as high as $\Theta(n)$). Hence a new approach is needed to show $o(\log{n})$, in particular, $O(1)$ node-averaged round complexity.
\subsection{Our Contributions}

Our main contributions are positing the sleeping model and designing algorithms with {\em constant}-rounds node-averaged awake complexity algorithm for MIS in the model.


\begin{table*}[ht]
\small
\begin{tabular}{|c|c|c|c|}
\hline

\multirow{2}{*}{}                                                        & \multirow{2}{*}{\begin{tabular}[c]{@{}c@{}}Prior MIS algorithms\\ (e.g., Luby's \cite{Luby_1986, Alon_1986},\\ CRT \cite{Coppersmith_1989, Blelloch_2012, Fischer_2018}, etc.)\end{tabular}} & \multicolumn{2}{c|}{Our algorithms}                                                                                                                                                     \\ \cline{3-4} 
                                                                         &                                                                                                                                   & \begin{tabular}[c]{@{}c@{}}Algorithm \ref{alg:sleepingMIS}\\(\textsc{SleepingMIS})\end{tabular} & \begin{tabular}[c]{@{}c@{}}Algorithm \ref{algorithm-fast-sleeping}\\(\textsc{Fast-SleepingMIS})\end{tabular} \\ \hline
\begin{tabular}[c]{@{}c@{}}Node-averaged\\ awake complexity\end{tabular} & Not applicable.                                                                                                                   & $O(1)$                                                                                & $O(1)$                                                                                          \\ \hline
\begin{tabular}[c]{@{}c@{}}Worst-case\\ awake complexity\end{tabular}    & Not applicable.                                                                                                                   & $O(\log{n})$                                                                          & $O(\log{n})$                                                                                    \\ \hline
\begin{tabular}[c]{@{}c@{}}Worst-case\\ round complexity\end{tabular}    & $O(\log{n})$                                                                                                                      & $O(n^3)$                                                                              & $O(\log^{3.41}{n})$                                                                             \\ \hline
\begin{tabular}[c]{@{}c@{}}Node-averaged\\ round complexity\end{tabular} & \begin{tabular}[c]{@{}c@{}}The best known\\ bounds are $O(\log{n})$.\end{tabular}                                                 & $O(n^3)$                                                                              & $O(\log^{3.41}{n})$                                                                             \\ \hline

\end{tabular}
\vspace{5 mm}
\caption{A summary of the different complexity measures for the various distributed MIS algorithms}
\label{table-list-of-complexity-measures-for-the-different-algorithms}
\end{table*}


Our main result is a randomized distributed MIS algorithm in the {\em sleeping model} whose {\em (expected) node-averaged awake complexity is $O(1)$.} In particular, we present a randomized distributed algorithm (Algorithm \ref{algorithm-fast-sleeping}) that has $O(1)$-rounds expected node-averaged awake complexity and, with high probability, has $O(\log{n})$-round worst-case awake complexity, and $O(\log^{3.41}n)$ worst-case (traditional) round complexity (cf.\ Theorem \ref{theorem-all-the-complexity-measures-for-the-second-algorithm}). We refer to Table \ref{table-list-of-complexity-measures-for-the-different-algorithms} for a comparison of the results. Please also see Theorem \ref{theorem-all-the-complexity-measures-for-the-first-algorithm} and Theorem \ref{theorem-all-the-complexity-measures-for-the-second-algorithm}, respectively.

Our work is\onlyLong{ also} a step towards understanding whether a $O(1)$-round node-averaged algorithm is possible in the traditional model (without sleeping).
\subsection{Comparison with Related Work} \label{sec:related}

\onlyLong
{
    Much of the research in design and analysis of efficient distributed algorithms and proving lower bounds of such algorithms in the last three decades have focused on the worst-case round complexity. Recently, there have been a few works that have focused on studying various fundamental distributed computing problems under a node-average round complexity (in the traditional model).
}

The notion of node-averaged (or vertex-averaged) complexity for the {\em traditional model} was proposed by Feuilloley \cite{Feuilloley_2020} and further studied (with slight modifications) in Barenboim and Tzur \cite{Barenboim_2019}. The motivation for node-averaged complexity --- which also applies here --- is that it can better capture the performance of distributed algorithms vis-a-vis the resources expended by individual nodes \cite{Barenboim_2019}.

In Feuilloley's notion \cite{Feuilloley_2020}, a node's running time is counted only till it outputs (or commits its output); the node may still participate later  (e.g., can forward messages etc.) but the time after it decides its output is {\em not counted}. In other words, the node-averaged complexity is the average of the runtimes of the nodes, where a node's runtime is till it decides its output (though it may not have terminated). The average running time is the average of the running times under the above notion. This work \cite{Feuilloley_2020} studies the average time complexity of leader election and coloring algorithms on cycles and other specific sparse graphs. \onlyLong{For leader election, the paper shows gives an algorithm with $O(\log{n})$ node-average complexity (we note that the worst-case time complexity has a lower bound of $\Omega(n)$, even for randomized algorithms \cite{Kutten_2015_JACM}). For 3-coloring a cycle, the paper \cite{Feuilloley_2020} shows that the node-averaged complexity cannot be improved over the worst-case, i.e., it is $\Theta(\log^*{n})$.}

Following the work of Feuilloley \cite{Feuilloley_2020}, Barenboim and Tzur \cite{Barenboim_2019} address several fundamental problems under node-averaged round complexity. Their notion is somewhat different from that of Feuilloley --- in \cite{Barenboim_2019}, as soon as a node decides its output, it sends its output to its neighbors and  terminates (does not take any further part in the algorithm). This is arguably a more suitable version for real-world networks in light of what was discussed in the context of saving energy and other node resources. Our notion is the same as that of Barenboim and Tzur, but extended to apply to the more general sleeping model where time spent by nodes in the sleeping state (if any) is not counted.

\onlyLong
{
    Barenboim and Tzur show a number of results for node-coloring as well as for MIS, $(2\Delta-1)$-edge-coloring and maximal matching. Their bounds apply to general graphs, but depend on the arboricity of the graph. In particular, for MIS, $(2\Delta-1)$-edge-coloring and maximal matching, they show a deterministic algorithm with $O(a+\log^* n)$ node-average complexity, where $a$ is the arboricity of the graph (which can be $\Theta(n)$ in general).

    It is still not known whether one can obtain $O(1)$ (or even $o(\log{n})$) round node-averaged complexity for MIS in general graphs in the traditional model (without sleeping). In this paper, we answer this question in the affirmative in the sleeping model. Note that $(\Delta+1)$-coloring can be solved in $O(1)$ round node-averaged complexity in general graphs by using Luby's $(\Delta+1)$-coloring algorithm \cite{Luby_1993}, e.g., see the paper of Barenboim and Tzur \cite[Section $6.2$]{Barenboim_2019}; however, this does not imply any such bound for MIS. 
}

There is also another important distinction between the algorithms of this paper and those of two works discussed above of \cite{Feuilloley_2020, Barenboim_2019}: some of the algorithms in the above works assume the \textsf{Local} model (unbounded messages), whereas \textsf{Congest} model (small-sized messages) is assumed here.

\onlyLong
{
    Finally, we point out that in a {\em dynamic} network model the work of \cite{Censor_2016_PODC} analyzed the average time complexity of algorithms using amortized analysis. This model is dynamic where nodes and edges may be added or deleted from the graph and is different compared to the static setting studied in this paper which is the case with almost all prior works on MIS mentioned here.
}

The work of King et al.\ \cite{King_2011} uses a sleeping model similar to ours --- nodes can be in two states sleeping or awake (listening and/or sending) --- but their setting is different. They present an algorithm in this model to solve a reliable broadcast problem in an energy-efficient way where nodes are awake only a fraction of the time. Another different model studied in the literature for problems such as MIS is the \emph{beeping model} (see, e.g., \cite{Afek_2013}) where nodes can communicate (broadcast) to their neighbors by either beeping or not. Sleeping  is orthogonal to beeping and one can study a model that uses both.

\section{Challenges and High-level overview} \label{sec:idea}

Before we go to an overview of our algorithm, we discuss some of the challenges in obtaining an $O(1)$-round node-averaged complexity in the traditional model. As mentioned in Section \ref{sec:focus}, either prior distributed MIS algorithms have  $\Theta(\log{n})$ node-average complexity or it is not clear if their node-averaged complexity is (even) $o(\log n)$ for arbitrary graphs. A straightforward  way to show constant node-averaged round complexity is to argue that a constant fraction (on expectation, for randomized algorithms) of the nodes finish in every round; this can be shown to imply $O(1)$ node-averaged complexity. Indeed, this is the reason why Luby's (randomized) algorithm for $(\Delta+1$)-coloring \cite{Luby_1993} gives $O(1)$ node-averaged complexity (see \cite[Section $6.2$]{Barenboim_2019}). It is not clear whether this kind of property can be shown for existing MIS algorithms (see Section \ref{sec:focus}).

\onlyLong{Our main contribution is to show how one can design an algorithm with $O(1)$-rounds node-averaged {\em awake} complexity for MIS in the sleeping model while still having small worst-case running times (both in the sleeping and traditional models). This is non-trivial since it is not obvious how to take advantage of the sleeping model properly. }A key difficulty in showing an $o(\log{n})$ node-averaged awake complexity for the MIS problem in the sleeping model is that messages sent to a sleeping node are simply ignored; they are not received at a later point. Hence a sleeping node is  unable to know the status of its neighbors (even after it is awake, since the neighbors could be then sleeping or even finished). Another difficulty is that it is not clear when to wake up a sleeping node; and when it wakes up, its neighbors might be sleeping and it won't know their status. It could be costly (in terms of node-averaged awake complexity) to keep all of its neighbors awake for many rounds. Hence a new approach is needed to get constant node-averaged awake complexity.


\onlyLong{

\begin{figure*}[t]
\begin{center}


\begin{tikzpicture}[level/.style = {sibling distance = 72 mm / #1}]

\node[circle, draw](z){$1$, $29$}
	child {node [circle,draw] (a) {$2$, $14$}
    child {node [circle,draw] (b) {$3$, $7$}
        child {node [circle,draw] (d) {$4$, $4$}}
        child {node [circle,draw] (e) {$6$, $6$}}
      } 
    child {node [circle,draw] (g) {$9$, $13$}
    child {node [circle,draw] (d) {$10$, $10$}}
        child {node [circle,draw] (e) {$12$, $12$}}
    }
  }
  child {node [circle,draw] (j) {$16$, $28$}
    child {node [circle,draw] (k) {$17$, $21$}
    child {node [circle,draw] (d) {$18$, $18$}}
        child {node [circle,draw] (e) {$20$, $20$}}
    }
  child {node [circle,draw] (l) {$23$, $27$}
      child {node [circle,draw] (o) {$24$, $24$}}
      child {node [circle,draw] (p) {$26$, $26$}
    }
  }
};
\node[](labelz)[right = of z,xshift=-8mm]{\textbf{\large $G = A \cup S$, $M(G) = M(A) \cup M(B)$}};
\node[](labela)[above left = of a,xshift=8mm,yshift=-15mm]{\textbf{\large $A$}};
\node[](labelj)[above right = of j,xshift=-8mm,yshift=-15mm]{\textbf{\large $S \supset B$}};

\end{tikzpicture}


\end{center}
\vspace{5 mm}
\caption{\large \boldmath A sample recursion tree consisting of four levels; each tree vertex is labeled with two numbers --- the first of which denotes the time when the vertex is reached for the first time, while the second number denotes the time when computation finishes at that vertex.}
\label{figure-recursion-tree-1}
\end{figure*}
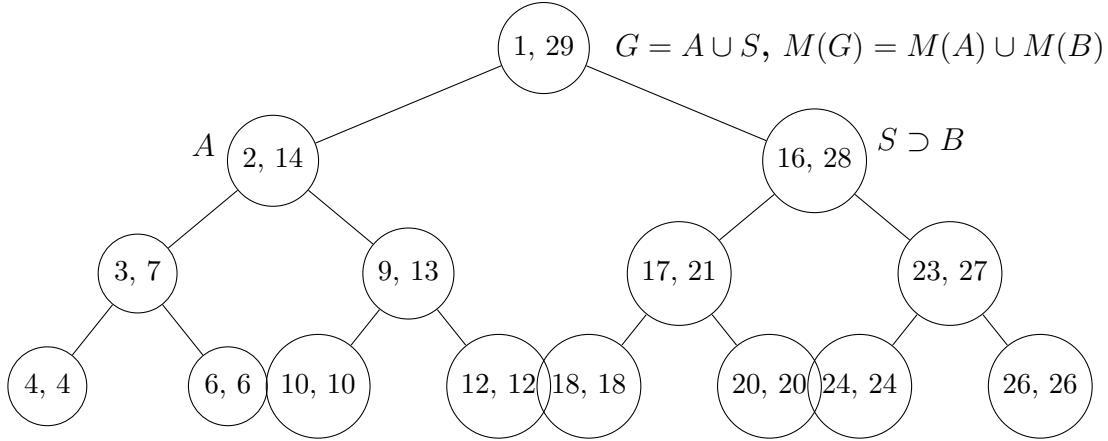
}


The high-level idea of our (randomized) algorithm is quite simple and can be explained by a simple recursive procedure\onlyLong{ (see Figure \ref{figure-recursion-tree-1})}. Consider a graph $G = (V, E)$. Every node flips a fair coin. If the coin comes up heads, the node falls asleep. Otherwise, it stays awake. Let $S$ be the set of sleeping nodes and $A$ be the set of awake nodes. The procedure is invoked recursively on the subgraph $G[A]$ induced by $A$ to compute an MIS $M(A)$ of that subgraph. The recursion bottoms out when the procedure is invoked on an empty subgraph or a subgraph containing only a single node. In the latter case, the node joins the MIS.

Once $M(A)$ is determined, the nodes in $S$ wake up at an appropriate time --- {\em that is synchronized to the time when the recursive call on $G[A]$ returns} -- and every node in $M(A)$ informs its neighbors that its in the MIS. At this point, the status (i.e., whether a node is in the MIS or not) of the nodes in $M(A)$ and their neighbors (including those in $S$) is fixed, so all of these nodes terminate.

It remains to fix the status of the remaining nodes, which form a subset $B \subseteq S$. To do so, we recursively invoke the procedure on $G[B]$, which gives us an MIS $M(B)$ of that subgraph. The overall MIS is then given by $M(A) \cup M(B)$. The recursion bottoms out when the procedure is invoked on an empty subgraph or a subgraph containing only a single node. In the latter case, the node joins the MIS.

The main observation for the analysis of this procedure is the following. By definition, $M(A)$ dominates $A$ and therefore all nodes in $A$ terminate after the first recursive call is finished. On top of that, the nodes in $M(A)$ might also dominate some of the nodes in $S$. In fact, one can show that on {\em on expectation at least a $1/4$-fraction} of the nodes in $S$ have a neighbor in $M(A)$ and hence will be {\em eliminated} when the recursive call finishes. This is shown in the key technical lemma called the {\em Pruning Lemma} (see Lemma \ref{lem:right}).

A main challenge in proving Lemma \ref{lem:right} is that $A$ is fixed by sampling (coin tosses) and the MIS of $A$ does not depend on $S$, the set of sleeping nodes. Yet we would like to show that a constant fraction of nodes in $S$ have a neighbor in the MIS of $A$, i.e., $M(A)$. Note that given $A$, $M(A)$ can possibly be such that the number of neighbors in $S$ can be very small; this will not eliminate many nodes in $S$. We avoid this by {\em coupling the process of sampling  with the process of finding an MIS} and show that, despite choosing $A$ first (by random sampling), the MIS computed on $A$ will eliminate a constant fraction of $S$. However, choosing $A$ first introduces dependencies which makes it non-trivial to prove the Pruning Lemma; we overcome this by using the {\em principle of deferred decisions} (cf.\ proof of Lemma \ref{lem:right}).

The Pruning Lemma (see Lemma \ref{lem:right}) guarantees that a constant fraction of nodes in the graph terminate {\em without being included} in either of the two recursive calls. So the status of these nodes is fixed by being awake for only a constant number of rounds (only three rounds). As a consequence, the two recursive calls together only operate on at most $\left(\frac{3}{4}\right)$-fraction of the given nodes on expectation. This saving propagates down the tree of recursive calls such that at level~$i$ of the tree the overall number of nodes on which the calls at that level operate is at most $\left(\frac{3}{4}\right)^i \cdot n$. It is not hard to see that the number of rounds per vertex required by the procedure outside of the recursive calls is constant. Therefore, the overall expected vertex-average complexity of the procedure in the sleeping model is: $\frac{1}{n} \sum_{i = 0}^\infty \left(\frac{3}{4} \right)^i n  =  O(1)$.

The above algorithm (cf.\ Section \ref{sec:alg}) has {\em constant} round node-averaged complexity and $O(\log{n})$-rounds worst-case awake complexity, but it has polynomial worst-case complexity. We then show that our MIS algorithm can be combined with a variant of Luby's algorithm so that the worst-case (traditional) complexity can be improved to polylogarithmic ($O(\log^{3.41}{n})$) rounds, while still having $O(1)$ rounds node-averaged complexity and $O(\log{n})$ worst-case awake complexity.

\section{The Sleeping MIS Algorithm} \label{sec:alg}

We consider the algorithm given in Figure \ref{alg:sleepingMIS}.
To compute an MIS for a given graph $G = (V, E)$, each node $v$ in $G$ calls the function \textsc{SleepingMIS} at the same time. The function \textsc{SleepingMIS} is called with the function parameter  $K = 3\log{n}$, where $n$ is the network size. We show in the analysis (see Lemma \ref{lem:correctness} and Theorem \ref{theorem-all-the-complexity-measures-for-the-first-algorithm}) that our algorithm is correct with high probability.

After initializing some variables, the function calls the recursive function \textsc{SleepingMISRecursive}. This function computes an MIS on the subgraph induced by the set of nodes that call it. In the initial call, all nodes in $G$ call the function. In later calls, however, the function operates on proper subgraphs of $G$. Each (non-trivial) call of \textsc{SleepingMISRecursive} partitions the set of nodes participating in the call into two subsets. The function uses a recursive call to compute an MIS on the subgraph induced by the first set, updates the nodes in the second set about the result of the first recursive call, and finally uses a second recursive call to finish the computation of the MIS. \textsc{SleepingMISRecursive} takes an integer parameter $k$. This parameter starts with $k  =  K  =  \lceil 3\log{n} \rceil$ in the initial call. The function parameter $k$ is then decremented from one level of the recursion to the next until the recursion base with $k = 0$ is reached.


\begin{algorithm}
  \begin{algorithmic}[1]

    \Function{SleepingMIS}{$K$}
    
      \State $v.$inMIS $\gets$ \texttt{unknown}
      \For{$i$ \textbf{in} $1, \dots, K$}
        \State $v.X_i \gets$ random bit that is $1$ with probability $\frac{1}{2}$
      \EndFor
      \State \textsc{SleepingMISRecursive}($K$) \label{line-depth-of-recursion} \Comment{$K$ is the recursion depth.}
    \EndFunction

    \bigskip

    \Function{SleepingMISRecursive}{$k$}
      \If{$k = 0$} \Comment{base case} \label{lin:baseCaseBegin}
        \State $v.$inMIS $\leftarrow$ true
        \State \textbf{return}
      \EndIf \label{lin:baseCaseEnd}

      \medskip
      \State send message to every neighbor \Comment{first isolated node detection, 1 round} \label{lin:isolatedBegin}
      \If{$v$ receives no message}
        \State $v.$inMIS $\leftarrow$ true
      \EndIf \label{lin:isolatedEnd}
       
      \medskip
      \If{$v.$inMIS $=$ unkown \textbf{and} $v.X_k = 1$} \Comment{left recursion} \label{lin:leftRecursionBegin}
        \State \textsc{SleepingMISRecursive}$(k - 1)$ \label{lin:leftRecursiveCall}
      \Else
        \State sleep for $T(k-1) = 3(2^{k-1}-1)$ rounds \Comment{$T(k-1)$ is the duration of the recursive call in Line~\ref{lin:leftRecursiveCall}\onlyLong{; it is computed in Section \ref{subsection-analysis-of-the-worst-case-awake-complexity-of-the-first-algorithm}}.} \label{lin:leftSleep}
      \EndIf \label{lin:leftRecursionEnd}

      \medskip
      \State send value of $v.$inMIS to every neighbor  \Comment{synchronization step,  1 round} \label{lin:eliminationBegin}
      \If{$v.$inMIS $=$ unknown \textbf{and} $v$ receives message from neighbor $w$ with $w$.inMIS $=$ true}
        \State $v.$inMIS $\leftarrow$ false \Comment{elimination}
      \EndIf \label{lin:eliminationEnd}

      \medskip
      \State send value of $v.$inMIS to every neighbor \Comment{second isolated node detection, 1 round} \label{lin:isolated2Begin}
      \If{$v.$inMIS $=$ unknown \textbf{and} $v$ only receives messages from neighbors $w$ with $w.\text{inMIS} = \text{false}$}
        \State $v.$inMIS $\leftarrow$ true
      \EndIf \label{lin:isolated2End}

      \medskip
      \If{$v.$inMIS $=$ unknown} \Comment{right recursion} \label{lin:rightRecursionBegin}
        \State \textsc{SleepingMISRecursive}$(k - 1)$ \label{lin:rightRecursiveCall}
      \Else
        \State sleep for $T(k-1) = 3(2^{k-1}-1)$ rounds \Comment{ $T(k-1)$ is the duration of the  recursive call in Line~\ref{lin:rightRecursiveCall}} 
        \label{lin:rightSleep}
      \EndIf \label{lin:rightRecursionEnd}
    \EndFunction
\end{algorithmic}

\caption{The ``Sleeping MIS algorithm'' executed by a node $v$} \label{alg:sleepingMIS}
\end{algorithm}


During the algorithm, each node $v$ in $G$ stores a variable $v.\text{inMIS}$. Initially, this variable is set to \emph{unknown} to signify that it has not yet been determined whether $v$ is in the MIS or not. Over the course of the algorithm the variable is set to \emph{true} or \emph{false}. Once $v.\text{inMIS}$ has been set to one of these values, it is never changed again. The MIS $M$ computed by the algorithm is given by the set of nodes $v$ with $v.\text{inMIS} = \text{true}$ after termination.

Consider a call of the function \textsc{SleepingMISRecursive} by a node set $U \subseteq V$ and with a parameter $k$. The goal of such a call is to compute an MIS in the induced subgraph $G[U]$. The function consists of six parts as indicated by the comments on the right side in Algorithm~\ref{alg:sleepingMIS}. The first part (Lines \ref{lin:baseCaseBegin} -- \ref{lin:baseCaseEnd}) is the \emph{base case} of the recursion. Once the base case of the recursion is reached it holds, with high probability, that $|U| \leq 1$. Therefore, the function has to compute an MIS on a graph consisting of at most one node. Such a node simply joins the MIS.

The second part (Lines~\ref{lin:isolatedBegin}--\ref{lin:isolatedEnd}) detects isolated nodes and adds them to the MIS. We refer to this part as the \emph{first isolated node detection}. Note that on the top level of the recursion, which operates on the entire graph $G$, the nodes do not have to communicate in order to determine whether a node is isolated or not. On lower levels of the recursion, however, the function generally operates on a proper subgraph $G[U]$ of $G$. The given instructions make sure that a node correctly determines its neighborhood in $G[U]$.

The third part (Lines~\ref{lin:leftRecursionBegin}--\ref{lin:leftRecursionEnd}) uses a recursive call to compute a partial MIS. We call this part the \emph{left recursion} due to the intuition of organizing the recursive calls of the function into a binary tree that is traversed in a left-to-right order (\onlyLong{see Figure \ref{figure-recursion-tree-1}}\onlyShort{See Figure 1 in the full paper in Appendix}). Every non-isolated node $v \in U$ with $v.X_k = 1$ participates in this recursive call, where $v.X_k$ is a variable that is set to a random bit during initialization. The recursive call computes an MIS in the subgraph induced by these nodes. In doing so, it fixes the value of $v.\text{inMIS}$ for every node $v \in U$ with $v.X_k = 1$. The nodes $v \in U$ with $v.X_k = 0$ and all isolated nodes in $U$ sleep for the duration of the left recursive call, see Line~\ref{lin:leftSleep}. Thereby, all nodes in $U$ start executing the next part of the function (which begins in Line~\ref{lin:eliminationBegin}) at the same time.

The purpose of the fourth and fifth part of the function is to update the nodes $v \in U$ with $v.X_k = 0$ about the decisions made in the left recursive call. We call the fourth part (Lines \ref{lin:eliminationBegin}--\ref{lin:eliminationEnd}) the \emph{elimination step}. In this part, every node $v \in U$ with $v.\text{inMIS} = \text{unkown}$ checks whether it has a neighbor in $G[U]$ that is in the MIS. If that is the case, $v$ sets $v.\text{inMIS}$ to false. The fifth part (Lines \ref{lin:isolated2Begin} -- \ref{lin:isolated2End}) is the \emph{second isolated node detection}. In this part, a node $v \in U$ checks whether the variable $w.\text{inMIS}$ is false for every neighbor $w$ of $v$ in $G[U]$. If so, $v$ sets $v.\text{inMIS}$ to true.

The sixth and final part of the function (Lines \ref{lin:rightRecursionBegin}--\ref{lin:rightRecursionEnd}) uses a second recursive call to complete the computation of the MIS in $G[U]$. We refer to this part as the \emph{right recursion}. As in the left recursion, only a subset of the nodes participates in the recursive call while the other nodes sleep. Specifically, every node $v \in U$ for which the value of $v.\text{inMIS}$ is still \emph{unknown} participates in the right recursive call.

One important technical issue is \emph{synchronization}. The sleeping nodes wake up at the appropriately synchronized round to synchronize with the wake up nodes (at the end of their recursive calls --- (Lines  \ref{lin:leftRecursiveCall} and \ref{lin:rightRecursiveCall}) so that messages can be exchanged between neighbors.\onlyLong{ A node that calls \textsc{SleepingMISRecursive}($k$) will wake up at a later round, which is given by the function $T(k)$. $T(k)$ is the ({\em worst-case}) time taken to complete this recursive call. The function $T(k) = 3(2^k - 1)$ as computed in the proof of Lemma \ref{lemma-first-algorithm-worst-case-traditional-round-complexity}.}

\section{Analysis} \label{sec:analysis}

We now turn to the analysis of  Algorithm \ref{alg:sleepingMIS}. \onlyShort{We only state the correctness lemma here and defer the formal proof of correctness to the longer version of the paper in the appendix.

\begin{lemma} \label{lem:correctness}
    It holds with high probability that on an $n$-node graph $G$, the set $M$ computed by the algorithm \textsc{SleepingMIS} is an MIS.
\end{lemma}
}

\onlyLong{\subsection{Correctness} \label{sec:correctness}

\begin{lemma} \label{lem:correctness}
   On an $n$-node graph $G$, the set $M$ computed by the algorithm \textsc{SleepingMIS}(K), with $K = \lceil3\log{n}\rceil$, is an MIS with high probability.
\end{lemma}
\begin{proof}
    We show the statement by induction. We use the following induction hypothesis: If all nodes in a set $U \subseteq V$ simultaneously call the function \textsc{SleepingMIS} with parameter $k$ then
    \begin{enumerate}
        \item all nodes in $U$ return from the call in the same round,
        \item after the call, the variable inMIS is true or false for all nodes in $U$, and
        \item the set of nodes in $U$ with inMIS $=$ true is an MIS of the subgraph induced by $U$.
    \end{enumerate}

    We begin by showing the induction step, i.e., we show that the above statement holds for node set $U$ and parameter $k > 0$ under the assumption that it holds for any subset of $U$ and parameter $k - 1$.
    
    For Condition~1 of the statement observe that outside of the recursive parts, all nodes perform the same instructions and, therefore, spend the same number of rounds. In the recursive parts, the nodes that participate in the recursive calls all return in the same round according to the induction hypothesis, and the nodes that do not participate in the recursive call sleep for the exact number of rounds required for the recursive call. Therefore, Condition~1 holds.

    To see that Condition~2 holds, observe that in Line~\ref{lin:rightRecursionBegin} every node in $U$ with inMIS $=$ unknown performs a recursive call so that the value of inMIS is set to true or false for all of these nodes according to the induction hypothesis.

    It remains to show that Condition~3 holds, i.e., the algorithm actually computes an MIS $M_{G[U]}$ on the subgraph $G[U]$ induced by $U$. Let $I \subseteq U$ be the set of isolated nodes in $G[U]$. The algorithm explicitly takes care of these nodes during the (first) isolated node detection (Lines \ref{lin:isolatedBegin} - \ref{lin:isolatedEnd}). The remaining nodes are partitioned into two sets $A$ and $B$ where $A$ is the set of nodes that participate in the first recursive call, and $B$ is the set of remaining non-isolated nodes.

    Finally, let $C \subseteq B$ be the set of nodes that participate in the second recursive call. To show that $M_{G[U]}$ is an MIS we first show that $M_{G[U]}$ is a \emph{dominating set} in $G[U]$ and then show that $M_{G[U]}$ is an \emph{independent set} in $G[U]$.

    To show that $M_{G[U]}$ is a dominating set in $G[U]$, we show that for every node $v \in U$ either $v$ is in $M_{G[U]}$ or a neighbor of $v$ is in $M_{G[U]}$. If $v \in I$ then $v$ joins the MIS during the first isolated node detection. If $v \in A$ then, according to the induction hypothesis, either $v$ joins the MIS during the first recursive call or some neighbor $u$ of $v$ that is also in $A$ joins the MIS.

    If $v \in B$ we have to distinguish between two subcases:

    If $v \notin C$ then we have two possibilities: (i) $v.\text{inMIS}$ was set to false during the elimination step (Lines \ref{lin:eliminationBegin} -- \ref{lin:eliminationEnd}) and, therefore, $v$ must have a neighbor $u$ that is in $A$ and joined the MIS; (ii) $v.\text{inMIS}$ was set to true during the   second isolated node detection (Lines \ref{lin:isolated2Begin}--\ref{lin:isolated2End}), and therefore is in the MIS (and dominated).

    Finally, the case $v \in C$ (Lines \ref{lin:rightRecursionBegin} -- \ref{lin:rightRecursiveCall}) is analogous to $v \in A$.

    To show that $M_{G[U]}$ is an independent set in $G[U]$, we prove that for any edge $\{ u, v \}$ in $G[U]$ if $u \in M_{G[U]}$ then $v \notin M_{G[U]}$. If $u, v \in A$ then the statement must hold according to the induction hypothesis. If $u \in A$ and $v \in B$ then $v.\text{inMIS}$ is set to false during the elimination step (Lines \ref{lin:eliminationBegin} -\ref{lin:eliminationEnd}) and the statement holds. If $u \in B$ and $v \in A$ then $u.\text{inMIS}$ was not set to false  (i.e., set to true) during  second isolated node detection (Lines \ref{lin:isolated2Begin} - \ref{lin:isolated2End}) and, therefore, we must have $v.\text{inMIS} = \text{false}$.

    For the case $u, v \in B$ we have to consider two subcases:

    If $u \notin C$ or $v \notin C$ then the statement holds due to second isolated node detection step (Lines \ref{lin:isolated2Begin} - \ref{lin:isolated2End}).


    Otherwise (i.e..both $u$ and $v$ are in $C$), the statement holds by the induction hypothesis.

    Finally, we consider the base case of the induction which corresponds to the base case of the recursion where $k = 0$. Recall that at the top level of the recursion, the function \textsc{SleepingMISRecursive} is called with the parameter $k$ set to $\lceil 3\log{n} \rceil$, and then $k$ is decremented from one level of the recursion to the next. Consider the recursion tree corresponding to the execution of the algorithm. Each tree node corresponds to a call of the function \textsc{SleepingMISRecursive} by a subset of the nodes in $G$. The root of the tree corresponds to the initial call made by all nodes in $G$. At an internal tree node, the probability that a non-isolated node in $G$ participates in the first recursive call is $1 / 2$, and the probability that a node in $G$ participates in the second recursive call is at most $1 / 2$. An isolated node does not participate in any recursive calls. Therefore, the probability that a node in $G$ participates in a specific call of \textsc{SleepingMISRecursive} corresponding to a leaf of the recursion tree is at most $2^{-k} = O(2^{-3\log n}) =  O(n^{-3})$. So for any pair of nodes $u, v$ in $G$ the probability that both nodes participate in the same recursive call at the bottom level of the recursion is at most $O(n^{-3})$. Applying the union bound over all pairs of nodes implies that, with high probability, at most one node participates in any given recursive call at the bottom level of the recursion. Therefore, if a node $v$ reaches the base case of the recursion then the function \textsc{SleepingMISRecursive} operates on a subgraph of $G$ containing only $v$, w.h.p. The node $v$ simply joins the MIS (see Lines~\ref{lin:baseCaseBegin}--\ref{lin:baseCaseEnd}). It is easy to check that, thereby, all three conditions of induction hypothesis are satisfied.
\end{proof}}

\subsection{Time (or Round) Complexity} \label{sec:timeComplexity}

The intuition behind the node-averaged running time analysis is that in each recursive call of \textsc{SleepingMISRecursive}, a constant fraction (on expectation) of the nodes calling the function participate in neither the left nor the right recursive call. Thereby, these nodes sleep for almost the entire duration of the call. This effect propagates through the levels of recursion and ultimately leads to a low node-averaged awake complexity.

To formalize this intuition, consider the execution of \textsc{SleepingMISRecursive} by a set of nodes $U \subseteq V$ and with parameter $k \ge 1$. Let $L \subseteq U$ be the set of nodes that participate in the left recursive call, and let $R \subseteq U$ be the set of nodes that participate in the right recursive call. The following lemma bounds the number of nodes participating in the left recursive call.

\begin{lemma} \label{lem:left}
    $\E \left[ \; |L| \; \middle| \; U \; \right] \le \frac{|U|}{2}$.
\end{lemma}
\begin{proof}
    Consider a node $v \in U$. If $v$ is isolated in $G[U]$, it joins the MIS in the first isolated node detection and, thereby, does not participate in the left recursive call. Otherwise, $v$ partipates in the left recursive call if and only if $v.X_k = 1$, which holds with probability $1/2$. The lemma follows by the linearity of expectation.
\end{proof}

The next is a key  lemma that establishes a bound for the right recursive call. The proof of this lemma requires a series of definitions and auxiliary lemmas.
\begin{lemma}[Pruning Lemma] \label{lem:right}
    $\E \left[ \; |R| \; \middle| \; U \; \right] \le \frac{|U|}{4}$.
\end{lemma}

Together, Lemmas~\ref{lem:left} and~\ref{lem:right} imply that the expected number of nodes that participate in either of the recursive calls is
\begin{center}
    $\E \left[ \; |L| + |R| \; \middle| \; U \; \right] \le \frac{3}{4} \cdot |U|$.
\end{center}
Thereby, on expectation, at least one-fourth of the nodes in $U$ do not participate in either of the recursive calls.

We now establish the definitions and lemmas necessary to prove Lemma~\ref{lem:right}. We recall from the algorithm that $K = \lceil 3\log{n} \rceil$ is the maximum recursion depth.

\begin{definition} \label{definition-k-rank}
    For any $k$ such that $0 \le k \le K$ and any node $v \in V$ we define the \emph{$k$-rank} of $v$ as the sequence
    \begin{center}
        $r_k(v) = (v.X_k, v.X_{k - 1}, \dots, v.X_1, -1)$.
    \end{center}    
    For two nodes $v, w$ we write $r_k(v) < r_k(w)$ if and only if $r_k(v)$ is lexicographically strictly less than $r_k(w)$. (Note that $r_0(v) = (-1)$, is simply a sentinel value for the base case.)
\end{definition}

Let $N_H(v)$ denotes the neighborhood of $v$ in the (sub)graph $H$. The following lemma provides us with a condition under which a node joins the MIS based on ranks.

\begin{lemma} \label{lem:maximalRank}
    Consider a call of \textsc{SleepingMISRecursive} by a node set $U$ with parameter $k$ during the execution of the algorithm. If $v \in U$ is such that every $w \in N_{G[U]}(v)$ with $r_k(w) > r_k(v)$ sets $w.\text{inMIS}$ to false at some point during the algorithm, then $v$ joins the MIS.
\end{lemma}

\onlyShort
{
    \begin{proof}[Proof Sketch.]
        The main idea is to use induction on $k$ and then do a case-by-case analysis depending on the value of the variable $v.X_k$. We defer the full proof to the long version of the paper (attached in the appendix) for the sake of preserving space here.
    \end{proof}
}

\onlyLong
{
\begin{proof}
    We prove the lemma by \emph{induction} on $k$.

    The \emph{induction base} with $k = 0$ corresponds to the base case of the recursion and holds trivially because if $k = 0$, then \textsc{SleepingMISRecursive} sets $v$ to be in MIS (Lines \ref{lin:baseCaseBegin} -- \ref{lin:baseCaseEnd}) regardless of rank value.

    For the \emph{induction step} we assume that the statement holds for $k - 1$ and show that it holds for $k$. Consider a call of \textsc{SleepingMISRecursive} by a node set $U$ with parameter $k$, and let $v \in U$ be such that every $w \in N_{G[U]}(v)$ with $r_k(w) > r_k(v)$ sets $w.\text{inMIS}$ to false at some point during the algorithm. If $v$ is isolated in $G[U]$ then $v$ joins the MIS in the first isolated node detection. Otherwise, we distinguish two cases based on the value of $v.X_k$.

    If $v.X_k = 1$ then $v$ participates in the left recursive call with parameter $k - 1$. Let $A$ be the set of $v$'s neighbors in $G[U]$ that also participate in the left recursive call. We show that each $w \in A$ with $r_{k - 1}(w) > r_{k - 1}(v)$ sets $w.\text{inMIS}$ to false at some point.

    Consider a node $w \in A$. If $r_k(w) \le r_k(v)$ then it holds $r_{k - 1}(w) \le r_{k - 1}(v)$ since $w.X_k = v.X_k = 1$. If $r_k(w) > r_k(v)$ then $w$ sets $w.\text{inMIS}$ to false at some point by definition. Thereby, the induction hypothesis implies that $v$ joins the MIS.

    If $v.X_k = 0$ then $v$ sleeps during the left recursive call. A node $w \in N_{G[U]}(v)$ with $r_k(w) \le r_k(v)$ must have $w.X_k = 0$. Thereby, $w$ also sleeps during the left recursive call, and it cannot cause $v$ to set $v.\text{inMIS}$ to false in the synchronization step. A node $w \in N_{G[U]}(v)$ with $r_k(w) > r_k(v)$ never sets $w.\text{inMIS}$ to true by definition. Thereby, such a node also cannot cause $v$ to set $v.\text{inMIS}$ to false in the synchronization step. If after the synchronization step it holds $w.\text{inMIS} = \text{false}$ for every $w \in N_{G[U]}(v)$, then $v$ joins the MIS in the second isolated node detection.

    Otherwise, $v$ participates in the right recursive call. Let $A$ be the set of $v$'s neighbors in $G[U]$ that also participate in the right recursive call. We show that each $w \in A$ with $r_{k - 1}(w) > r_{k - 1}(v)$ sets $w.\text{inMIS}$ to false at some point.

    Consider a node $w \in A$. If $r_k(w) \le r_k(v)$ then it holds $r_{k - 1}(w) \le r_{k - 1}(v)$ since $w.X_k = v.X_k = 0$. If $r_k(w) > r_k(v)$ then $w$ sets $w.\text{inMIS}$ to false at some point by definition. Thereby, the induction hypothesis implies that $v$ joins the MIS.
\end{proof}
}

\begin{definition} \label{definition-evaluation-sequence}
    Consider a call of \textsc{SleepingMISRecursive} by a node set $U \subseteq V$ with parameter $k$ during the execution of the algorithm. Let $(v_1, \dots, v_{|U|})$ be the sequence of the nodes in $U$ sorted by lexicographically decreasing $(k - 1)$-rank. We refer to this sequence as the \emph{evaluation sequence}.
\end{definition}

\begin{remark}
    Note that the order of the nodes in the evaluation sequence is defined based on the $(k - 1)$-rank of the nodes and not the $k$-rank. Intuitively, the $(k - 1)$-rank of a node is the sequence of random decisions used in \emph{future} recursions. The following lemma shows that the evaluation sequence is well defined.
\end{remark}

\begin{lemma} \label{lem:distinctRanks}
    For every call of \textsc{SleepingMISRecursive} by a node set $U$ with parameter $k$ during the execution of the algorithm and for any $v, w \in U$, it holds with high probability that $r_{k - 1}(v) \neq r_{k - 1}(w)$.
\end{lemma}
\begin{proof}
    For two nodes $v, w \in V$ to participate in the same call with node set $U$ and parameter $k$, we must have $v.X_i = w.X_i$ for all $i$ such that $k < i \le K$. For both nodes to have the same $(k - 1)$-rank, we must have $v.X_i = w.X_i$ for all $i$ such that $1 \le i < k$. Therefore, we have
    \begin{center}
        $\Pr[\, v, w \in U \text{ and } r_{k - 1}(v) = r_{k - 1}(w)\, ] = 2^{-{K - 1}} \le 2n^{-3}$.
    \end{center}
    The statement holds by applying the union bound over all pairs of nodes.
\end{proof}

\paragraph{Using the principle of ``deferred decisions''.}

The proof of Lemma~\ref{lem:right} is based on the \emph{principle of deferred decisions}\footnote{For an introduction to the principle see e.g., \cite{upfal}.}: We still focus on a call of \textsc{SleepingMISRecursive} by a node set $U$ with parameter $k$ during the execution of the algorithm. We assume for the sake of analysis that the values of the variables $X_k$ are not fixed at the beginning of the algorithm. However, the remaining $X_i$ for $i \neq k$ are fixed, which means that the set $U$, the $(k - 1)$-ranks of the nodes, and the evaluation sequence $(v_1, \dots, v_{|U|})$ have all been determined. We then fix the values of the $X_k$ using the following deferred decision process.


\begin{algorithm}
    \begin{enumerate}
        \item Let $v_i$ be the node in the evaluation sequence with the smallest $i$ such that $v_i.X_k$ has not yet been fixed. We fix the value of $v_i.X_k$ by flipping a fair coin and we say that $v_i$ is \emph{sequence-fixed}.
        
        \item If $v_i.X_k = 1$, we also fix the value of $w.X_k$ for each neighbor $w \in N_{G[U]}(v_i)$ for which $w.X_k$ has not yet been fixed, again by flipping independent and fair coins. We say that such a neighbor $w$ is \emph{neighbor-fixed}.
        
        \item We repeat these steps until the variable $X_k$ is fixed for all nodes in the evaluation sequence.
    \end{enumerate}

  \caption*{Deferred Decision Process}
\end{algorithm}


Note that the above process does \emph{not} change the behavior of the algorithm. It merely serves as a useful tool that allows us to analyze the algorithm as demonstrated in the following lemma.

\begin{lemma} \label{lem:connection}
  The following statements hold for every node $v \in U$.
  \begin{enumerate}
    \item \label{case:sequenceFixed}
    If $v$ is sequence-fixed and $v.X_k = 1$,
    then $v$ sets $v.\text{inMIS}$ to true before the synchronization step.

    \item \label{case:neighborFixed}
    If $v$ is neighbor-fixed,
    then $v$ sets $v.\text{inMIS}$ to false before the second isolated node detection.
  \end{enumerate}
\end{lemma}
\begin{proof}
  We show the lemma by complete induction on the evaluation sequence $(v_1, \dots, v_{|U|})$.
  The node $v_1$ is sequence-fixed by definition.
  If $v_1.X_k = 0$, none of the statements apply and the lemma holds.
  If $v_1.X_k = 1$, we have to show that Statement~\ref{case:sequenceFixed} holds.
  By definition of the evaluation sequence we have $r_{k - 1}(v_1) > r_{k - 1}(v_j)$ for all $j > 1$.
  Together with $v_1.X_k = 1$ this implies that $r_k(v_1) > r_k(v_j)$ for all $j > 1$.
  Therefore, Lemma~\ref{lem:maximalRank} implies that $v_1$ joins the MIS.
  Since $v_1.X_k = 1$, $v_1$ must set $v_1.\text{inMIS}$ to true in the first isolated node detection or the left recursive call.
  Thereby, Statement~\ref{case:sequenceFixed} and the induction base hold.

  For the induction step we consider an index $i \ge 2$ and show that the lemma holds for $v_i$ under the assumption
  that it holds for all $v_j$ with $j < i$.
  We first show that Statement~\ref{case:sequenceFixed} holds.
  Suppose that $v_i$ is sequence-fixed and $v_i.X_k = 1$.
  We have to show that $v_i$ sets $v_i.\text{inMIS}$ to true before the synchronization step.
  Consider a node $w \in N_{G[U]}(v_i)$.
  If $r_{k - 1}(w) \le r_{k - 1}(v_i)$ then we must also have $r_k(w) \le r_k(v_i)$ because $v_i.X_k = 1$.
  If $r_{k - 1}(w) > r_{k - 1}(v_i)$ then the lemma holds for $w$ according to the induction hypothesis.
  We distinguish two cases.
  If $w$ is sequence-fixed then we must have $w.X_k = 0$ because otherwise $v_i$ would be neighbor-fixed, which is a contradiction.
  Therefore, we have $r_k(w) \le r_k(v_i)$ in this case.
  If $w$ is neighbor-fixed then $w$ sets $w.\text{inMIS}$ to false before the second isolated node detection
  according to Statement~\ref{case:neighborFixed}.
  In summary, we have that for every node $w \in N_{G[U]}(v_i)$ it holds $r_k(w) \le r_k(v_i)$
  or $w$ sets $w.\text{inMIS}$ to false at some point during the algorithm.
  Thereby, Lemma~\ref{lem:maximalRank} implies that $v_1$ joins the MIS.
  Since $v_1.X_k = 1$, $v_1$ must set $v_1.\text{inMIS}$ to true in the first isolated node detection or the left recursive call.
  Therefore, Statement~\ref{case:sequenceFixed} holds.

  Finally, we show that Statement~\ref{case:neighborFixed} holds.
  Suppose that $v_i$ is neighbor-fixed.
  We have to show that $v_i$ sets $v_i.\text{inMIS}$ to false before the second isolated node detection.
  Since $v_i$ is neighbor-fixed, it must have a neighbor $w \in N_{G[U]}(v_i)$ such that
  $r_{k - 1}(w) > r_{k - 1}(v_i)$, $w$ is sequence-fixed, and $w.X_k = 1$.
  Since $r_{k - 1}(w) > r_{k - 1}(v_i)$ the lemma holds for $w$ according to the induction hypothesis.
  Because $w$ is sequence-fixed and $w.X_k = 1$, Statement~\ref{case:sequenceFixed} implies
  that $w$ sets $w.\text{inMIS}$ to true before the synchronization step.
  Since the algorithm is correct (see Lemma \ref{lem:correctness}), this implies that $v_i$ must eventually set $v_i.\text{inMIS}$ to false.
  If $v_i.X_k = 0$ then $v_i$ must set $v_i.\text{inMIS}$ to false in the left recursive call.
  Otherwise, $v_i$ learns during the synchronization step that $w$ set $w.\text{inMIS}$ to true
  and, therefore, $v_i$ sets $v_i.\text{inMIS}$ to false in the synchronization step.
  Thereby, Statement~\ref{case:neighborFixed} holds.
\end{proof}

We are finally ready to prove Lemma~\ref{lem:right}.

\begin{proof}[Proof of Lemma~\ref{lem:right}]
  Consider a call of \textsc{SleepingMISRecursive} by a node set $U$ with parameter $k$ during the execution of the algorithm.
  We show that for each $v \in U$ it holds $\Pr[ \, v \in R \, ] \le 1 / 4$.
  The linearity of expectation then implies $\E \left[ \; |R| \; \middle| \; U \; \right] \le |U| / 4$, so the lemma holds.

  If $v$ is isolated in $G[U]$, then $v$ joins the MIS in the first isolated node detection
  and we have $\Pr[ \, v \in R \, ] = 0$.
  Thereby, the statement holds for isolated nodes.
  If $v$ is not isolated, we can apply the \emph{law of total probability} to obtain
  \begin{align*}
    \Pr[ \, v \in R \, ]
    = & \Pr[ \, v \in R \, | \, v.X_k = 1 \, ] \cdot \Pr[ \, v.X_k = 1 \, ] \\
    + & \Pr[ \, v \in R \, | \, v.X_k = 0 \, ] \cdot \Pr[ \, v.X_k = 0 \, ].
  \end{align*}
  By definition, it holds $\Pr[ \, v \in R \, | \, v.X_k = 1 \, ] = 0$ and $\Pr[ \, v.X_k = 0 \, ] = 1 / 2$.
  Therefore, we have
  \[
    \Pr[ \, v \in R \, ]
    = \frac{1}{2} \cdot \Pr[ \, v \in R \, | \, v.X_k = 0 \, ].
  \]

  We define $E$ as the event that all neighbors in $N_{G[U]}(v)$ are neighbor-fixed.
  Recall that $v$ is not isolated.
  We apply the \emph{law of total probability} again to obtain
  \begin{align*}
    &\Pr[ \, v \in R \, | \, v.X_k = 0 \, ]\\
    &=  \Pr[ \, v \in R \, | \, v.X_k = 0 \, \wedge \, E \, ] \cdot \Pr[ \, v.X_k = 0 \, \wedge \, E \,] \\
    &+  \Pr[ \, v \in R \, | \, v.X_k = 0 \, \wedge \, \overline{E} \, ] \cdot \Pr[ \, v.X_k = 0 \, \wedge \, \overline{E} \,].
  \end{align*}

  We first determine the probability $\Pr[ \, v \in R \, | \, v.X_k = 0 \, \wedge \, E \, ]$.
  If $v.X_k = 0$, $v$ cannot cause any of its neighbors to be neighbor-fixed.
  Hence, if the event $E$ occurs then every node in $N_{G[U]}(v)$ must have been neighbor-fixed by some other node.
  In this case, Lemma~\ref{lem:connection} implies that every node $w \in N_{G[U]}(v)$ sets $w.\text{inMIS}$ to false
  before the second isolated node detection.
  Hence, $v$ sets $v.\text{inMIS}$ to true during the second isolated node detection.
  This implies that $\Pr[ \, v \in R \, | \, v.X_k = 0 \, \wedge \, E \, ] = 0$.

  Next, we establish an upper bound on $\Pr[ \, v \in R \, | \, v.X_k = 0 \, \wedge \, \overline{E} \, ]$.
  If $\overline{E}$ occurs, then there is a node $w \in N_{G[U]}(v)$ that is sequence-fixed.
  With probability $1 / 2$, we have $w.X_k = 1$.
  In this case, Lemma~\ref{lem:connection} implies that $w$ sets $w.\text{inMIS}$ to true before the synchronization step.
  Thereby, $v$ sets $v.\text{inMIS}$ to false during the synchronization step.
  This implies $\Pr[ \, v \in R \, | \, v.X_k = 0 \, \wedge \, \overline{E} \, ] \le 1 / 2$.

  By combining the equations above, we get
  $\Pr[ \, v \in R \, | \, v.X_k = 0 \, ] \le 1 / 2$
  and, therefore,
  $\Pr[ \, v \in R \, ] \le 1 / 4$.
\end{proof}


\onlyLong{
\begin{lemma} \label{lem:levelCost}
    Consider all calls of \textsc{SleepingMISRecursive} with the same parameter $k$. Let $U_j$ be the set of nodes participating in the $j$-th of these calls. Note that all the $U_j$ are pairwise disjoint by definition. We define
        $$Z_k = \sum_j |U_j|$$
    to be the number of nodes participating in the recursive calls with parameter $k$. Then the following holds for all $i$ such that $0 \leq i \leq K$:
    \begin{center}
        $\E \left[ Z_{K - i} \right] \le \left( \frac{3}{4} \right)^i \cdot n$.
    \end{center}
\end{lemma}
}

\onlyShort{
\begin{lemma} \label{lem:levelCost}
    Consider all calls of \textsc{SleepingMISRecursive} with the same parameter $k$. Let $U_j$ be the set of nodes participating in the $j$-th of these calls. Note that all the $U_j$ are pairwise disjoint by definition. We define $Z_k = \sum_j |U_j|$ to be the number of nodes participating in the recursive calls with parameter $k$. Then the following holds for all $i$ such that $0 \leq i \leq K$: $\E \left[ Z_{K - i} \right] \le \left( \frac{3}{4} \right)^i \cdot n$.
\end{lemma}
}

\begin{proof}
  We show the lemma by induction on $i$.
  For the induction base consider $i = 0$.
  This case corresponds to the initial call of the recursive function by all nodes.
  Thus, we have $Z_K = n$ by definition, which shows the induction base.

  For the induction step we assume that the equation holds for an $i$ such that $0 \le i < K$ and show that it also holds for $i + 1$.
  Consider the calls of \textsc{SleepingMISRecursive} with parameter $k = K - i$ during the algorithm.
  Let $m$ be the number of these calls,
  and let $U_j$ be the set of nodes participating in the $j$-th call.
  By definition, we have
  \begin{equation} \label{eqn:Z_K-i}
    Z_{K - i} = \sum_{j = 1}^m |U_j|.
  \end{equation}
  Each of the $m$ calls creates a left and right recursive call at the next level of recursion.
  For the call $j$-th call let $L_j, R_j \subseteq U_j$ be the node sets that participate in the left and right recursive call, respectively.
  By definition, we have
  \[
    Z_{K - (i + 1)}
    = \sum_{j = 1}^m \left( |L_j| + |R_j| \right).
  \]
  This implies that
  \begin{align*}
    \E [ \, Z_{K - (i + 1)} \, | \, U_1, \dots, U_m \, ]
    &=  \E \left[ \sum_{j = 1}^m \, \left( |L_j| + |R_j| \right) \, \middle| \, U_1, \dots, U_m \, \right] \\
    &=  \sum_{j = 1}^m \E \left[ \, |L_j| + |R_j| \, \middle| \, U_j \, \right]\\
    &\leq   \frac{3}{4} \sum_{j = 1}^m |U_j|   =   \frac{3}{4} \cdot Z_{K - i}\text{,}
  \end{align*}
  where the inequality follows from Lemmas~\ref{lem:left} and~\ref{lem:right}, and the last equality follows from Equation~\ref{eqn:Z_K-i}. Taking the expectation on both sides of this equation and applying the induction hypothesis gives us
\begin{center}
    $\E[Z_{K - (i + 1)}]   \leq   \frac{3}{4} \cdot \E[Z_{K - i}]   \leq   (\frac{3}{4})^{i + 1} \cdot n$.
\end{center}
\end{proof}

\begin{lemma} \label{lemma-first-algorithm-node-averaged-awake-complexity}
    The expected node-averaged awake complexity of Algorithm~\ref{alg:sleepingMIS} in the sleeping model is $O(1)$.
\end{lemma}
\begin{proof}
    For each node $v \in V$, we define a cost $C(v)$ which is the number of rounds $v$ is awake during the execution of the algorithm. We define $C = \sum_{v \in V} C(v)$ to be the total cost of the algorithm. The expected node-average complexity is then $\frac{\E[C]}{n}$.

    The initialization at the beginning of the algorithm incurs a cost of $O(1)$ for each node. To compute the cost of the recursive part of the algorithm, we use the following observation: Consider a node $v$ and a call of \textsc{SleepingMISRecursive} that $v$ participates in. If we ignore the cost of the recursive calls, the participation of $v$ in the call incurs a cost of $O(1)$. This implies that we can compute the total cost of the recursive part by computing the number of nodes that participate in the recursive calls at the different levels of the recursion. Recall that we defined the number of nodes participating in the calls with common parameter $k$ as $Z_k$. Formally, we have
    \begin{center}
        $C = O(1) \cdot \sum_{k = 0}^K Z_k$.
    \end{center}    
    Taking the expectation on both sides and applying Lemma~\ref{lem:levelCost} gives us
    \begin{align*}
        \E[C]
        &=   O(1) \cdot \sum_{k = 0}^K \E[Z_k]   =   O(1) \cdot \sum_{i = 0}^K \E[Z_{K - i}]\\
        &\leq   O(1) \cdot \sum_{i = 0}^K \left( \frac{3}{4} \right)^i \cdot n   \leq   O(n) \cdot \sum_{i = 0}^\infty \left( \frac{3}{4} \right)^i   =   O(n).
    \end{align*}
    Therefore, the expected node-average awake complexity of the algorithm is $\frac{\E[C]}{n} = O(1)$.
\end{proof}
\onlyLong{\subsection{Analysis of the worst-case awake complexity of Algorithm \ref{alg:sleepingMIS}} \label{subsection-analysis-of-the-worst-case-awake-complexity-of-the-first-algorithm}

\begin{lemma} \label{lemma-first-algorithm-worst-case-awake-round-complexity}
    The worst-case awake round complexity of Algorithm \ref{alg:sleepingMIS} is $O(\log{n})$.
\end{lemma}
\begin{proof}
    The worst-case awake round complexity can be seen to be proportional to the depth of the recursion tree of the Algorithm \textsc{SleepingMISRecursive}($k$). We note that the activation of nodes proceeds in a depth-first, left-to-right fashion. Also, each node is awake only a constant number of rounds for possibly every level of the recursion tree. Thus, any node can be awake at most a number of rounds proportional to the depth of the recursion, which is $O(\log{n})$ rounds.
\end{proof}

\begin{lemma} \label{lemma-first-algorithm-worst-case-traditional-round-complexity}
    The worst-case round complexity of Algorithm \ref{alg:sleepingMIS} is $O(n^3)$. 
\end{lemma}
\begin{proof}
    The worst-case round complexity of Algorithm \textsc{SleepingMISRecursive}($k$) can be computed recursively. Let $T(k)$ be the worst-case number of rounds taken by \textsc{SleepingMISRecursive}($k$). Then
    \begin{center}
        $T(k)   \leq   2T(k - 1)  +  3$  for  $k > 0$,
    \end{center}    
    since \textsc{SleepingMISRecursive}($k$) calls \textsc{SleepingMISRecursive}($k - 1$) (at most) two times and takes 3 extra rounds. The base case is $T(0) = 0$. The above recurrence gives \onlyShort{$T(k)  \leq  3(2^{k} - 1)$.}
    \onlyLong
    {
        \begin{center}
            $T(k)  \leq  3(2^{k} - 1)$.
        \end{center}    
    }
    
    The depth of the tree is $3\log{n}$ (see Lines \ref{line-depth-of-recursion} and \ref{lin:baseCaseBegin} of the pseudocode of Algorithm \ref{alg:sleepingMIS}; also see Figure \ref{figure-recursion-tree-2}), and hence plugging this value of $k$ gives the worst-case round complexity as $O(n^3)$.
\end{proof}

As an immediate corollary to Lemma \ref{lemma-first-algorithm-worst-case-traditional-round-complexity}, we have
\begin{lemma} \label{lemma-first-algorithm-node-averaged-traditional-round-complexity}
    The node-averaged round complexity of Algorithm \ref{alg:sleepingMIS} is $O(n^3)$.
\end{lemma}}
\onlyShort{The analysis of the traditional (i.e., worst-case with regard to the individual nodes) round complexity of Algorithm \ref{alg:sleepingMIS} is deferred to the full version of the paper for the sake of preserving space here. Combining the results corresponding to the different complexity measures}\onlyLong{Combining Lemmas \ref{lemma-first-algorithm-node-averaged-awake-complexity}, \ref{lemma-first-algorithm-worst-case-awake-round-complexity}, \ref{lemma-first-algorithm-worst-case-traditional-round-complexity}, and \ref{lemma-first-algorithm-node-averaged-traditional-round-complexity}} yields the main result of this section.

\begin{theorem} \label{theorem-all-the-complexity-measures-for-the-first-algorithm}
    There is a randomized Monte Carlo,\footnote{We recall that a Monte Carlo randomized algorithm is one that may sometimes produce an incorrect solution. In contrast, Las Vegas algorithms are randomized algorithms that always produce the correct solution.\onlyLong{ Please refer to \cite[Section $1.2$]{Motwani_1995_Book} for a detailed discussion on these two classes of randomized algorithms.}} distributed MIS algorithm --- described in Algorithm \ref{alg:sleepingMIS} --- that is correct with high probability and has the following performance measures:
    \begin{itemize}
        \item $O(1)$ node-averaged awake complexity (on expectation),
        \item $O(\log{n})$ worst-case awake complexity (always),
        \item $O(n^3)$ node-averaged round complexity (always), and
        \item $O(n^3)$ worst-case round complexity (always).
    \end{itemize}
\end{theorem}

\subsection{Improving the Worst-case Round Complexity} \label{sec:worst-case}

We now show how to reduce the worst-case round complexity significantly to polylogarithmic rounds while still keeping the node-averaged awake complexity to be constant and the worst-case awake complexity to be $O(\log{n})$ rounds. In particular, we show the following theorem, which is the main result of this section.

\begin{theorem} \label{theorem-all-the-complexity-measures-for-the-second-algorithm}
    There is a randomized, Monte-Carlo distributed MIS algorithm --- described in Algorithm \ref{algorithm-fast-sleeping} --- that is correct with high probability and  has the following performance measures:
    \begin{itemize}
        \item $O(1)$ node-averaged awake complexity (on expectation),
        \item $O(\log{n})$ worst-case awake complexity (with high probability),
        \item $O(\log^{3.41}{n})$ node-averaged round complexity (with high probability), and
        \item $O(\log^{3.41}{n})$ worst-case round complexity (with high probability).
    \end{itemize}
\end{theorem}

\onlyLong{In the rest of the section, we prove Theorem \ref{theorem-all-the-complexity-measures-for-the-second-algorithm}.}

\paragraph{Modifications to Algorithm \textsc{SleepingMISRecursive}($k$).}

The main idea is to truncate the recursion tree of Algorithm \textsc{SleepingMISRecursive}($k$) earlier (see Figure \ref{figure-recursion-tree-2}) and use the \emph{parallel/distributed randomized greedy Maximal Independent Set (MIS) algorithm} (described below) to solve the base cases.

\paragraph{The parallel/distributed randomized greedy Maximal Independent Set (MIS) algorithm \cite{Coppersmith_1989, Blelloch_2012, Fischer_2018}.}

The parallel/distributed randomized greedy MIS algorithm works as follows\cite{Fischer_2018}. An \emph{order} (also called \emph{ranking}) of the vertices is chosen uniformly at random. Then, in each round, all vertices having the highest rank among their respective neighbors are added to the independent set and removed from the graph along with their neighbors. This process continues iteratively until the resulting graph is empty. \onlyShort{Fischer and Noever \cite{Fischer_2018} showed that the parallel/distributed randomized greedy MIS algorithm finishes in $O(\log{n})$ rounds with high probability.}

\onlyLong{
The parallel/distributed randomized greedy MIS algorithm was first introduced by Coppersmith et al.\ \cite{Coppersmith_1989}. They used this algorithm to find an MIS for $G(n, p)$ random graphs and showed that it runs in $O(\log^2{n})$ \emph{expected} rounds and that this holds true for all values of $p  \in  [0, 1]$. Blelloch et al.\ \cite{Blelloch_2012} extended this result to general (arbitrary) graphs and showed an $O(\log^2{n})$ run-time with high probability. Finally, Fischer and Noever \cite{Fischer_2018} improved the analysis further and showed that the parallel/distributed randomized greedy MIS algorithm ran in $O(\log{n})$ rounds with high probability and also that this bound was tight. Fischer and Noever's work showed that the parallel/distributed randomized greedy MIS algorithm was (asymptotically) as fast as the famous algorithm by Luby \cite{Luby_1986, Alon_1986}, where the random ranking of the nodes is chosen afresh for each iteration.
}

It is a well-known fact (see, e.g., \cite{Blelloch_2012}) that the parallel/distributed randomized greedy MIS algorithm always produces what is known as the \emph{lexicographically first} MIS \cite{Coppersmith_1989}, i.e., the same MIS output by the sequential greedy algorithm. In other words, it always produces the same result once an ordering of the vertices is fixed. We observe that Algorithm \textsc{SleepingMISRecursive}($k$) produces a lexicographically first MIS as well. This follows as an immediate corollary to Lemma \ref{lem:maximalRank}. Formally,

\begin{corollary} \label{cor:samemis}
    Algorithm \textsc{SleepingMISRecursive}($k$) and the parallel/distributed randomized greedy MIS algorithm produce the same MIS.
\end{corollary}


\onlyLong{\begin{figure*}[t]
\begin{center}


\begin{tikzpicture}[level/.style = {sibling distance = 25 mm / #1, level distance = 10mm}]

\node[circle, draw] (z){}
	child {node [circle, draw] (y){}
		child {node [circle, draw] (p) {}}
		child {node [circle, draw] (q) {}}
		}
	child {node [circle,draw] (x){}
		child {node [circle, draw] (r) {}}
		child {node [circle, draw] (s) {}}
		};

\node[] (p2k) [below left = of p, xshift=8mm, yshift = 2 mm] {\textbf{$\vdots$}};
\node[] (s2n) [below right = of s, xshift=-8mm, yshift = 2 mm] {\textbf{$\vdots$}};
\node[circle, draw] (k) [below left = of p2k, xshift=8mm, yshift = 2 mm] {};
\node[] (k2l) [right = of k, xshift=-8mm] {$\cdots$}; 
\node[circle, draw] (l) [right = of k2l, xshift=-7mm] {}; 
\node[] (l2m) [right = of l, xshift=-8mm] {$\cdots$};
\node[circle, draw] (m) [right = of l2m, xshift=-7mm] {};
\node[] (m2n) [right = of m, xshift=-8mm] {$\cdots$};
\node[circle, draw] (n) [below right = of s2n, xshift=-8mm, yshift = 2 mm] {}; 
\draw [dashed] (p) -- (p2k) -- (k);
\draw [dashed] (s) -- (s2n) -- (n);
\node[] (q2) [below right = of q, xshift=-8mm, yshift = 2 mm] {};
\node[] (r2) [below left = of r, xshift=8mm, yshift = 2 mm] {};
\draw [dashed] (q) -- (q2);
\draw [dashed] (r) -- (r2);

\node[] (k2a) [below left = of k, xshift=10mm, yshift = 2 mm] {\textbf{$\vdots$}};
\node[] (n2f) [below right = of n, xshift=-10mm, yshift = 2 mm] {\textbf{$\vdots$}};
\node[circle, draw] (a) [below left = of k2a, xshift=10mm, yshift = 2 mm] {};
\node[circle, draw] (b) [right = of a, xshift=-8mm] {};
\node[] (b2c) [right = of b, xshift=-8mm] {$\cdots$};
\node[circle, draw] (c) [right = of b2c, xshift=-6mm] {};
\node[circle, draw] (d) [right = of c, xshift=-8mm] {};
\node[circle, draw] (e) [right = of d, xshift=-8mm] {};
\node[] (e2f1) [right = of e, xshift=-6mm] {$\cdots$};
\node[] (e2f2) [right = of e2f1, xshift=-6mm] {$\cdots$};
\node[circle, draw] (f) [below right = of n2f, xshift=-8mm, yshift = 2 mm] {}; 
\draw [dashed] (k) -- (k2a) -- (a);
\draw [dashed] (n) -- (n2f) -- (f);

\node[] (a2g) [below left = of a, xshift=10mm, yshift = 2 mm] {\textbf{$\vdots$}};
\node[] (f2j) [below right = of f, xshift=-10mm, yshift = 2 mm] {\textbf{$\vdots$}};
\node[circle, draw, dashed] (g) [below left = of a2g, xshift=8mm, yshift = 2 mm] {};
\node[] (g2h) [right = of g, xshift=-3mm] {$\cdots$};
\node[circle, draw, dashed] (h) [right = of g2h] {};
\node[] (h2i) [right = of h, xshift=-1mm] {$\cdots$};
\node[circle, draw, dashed] (i) [right = of h2i, xshift=5mm] {};
\node[] (i2j) [right = of i] {$\cdots$};
\node[circle, draw, dashed] (j) [below right = of f2j, xshift=-8mm, yshift = 2 mm] {};
\draw [dashed] (a) -- (a2g) -- (g);
\draw [dashed] (f) -- (f2j) -- (j);


\node [text width=4cm, align=right] (label3) [right = of j, xshift=-16mm] {$\text{level} = 0$, \\$\text{depth} = K = c\log{n}$ \\(the base case for Algorithm \ref{alg:sleepingMIS})};
\node [text width=4cm, align=right] (label2) [right = of f, xshift=-9mm] {$\text{level} = K - \ell\log{\log{n}}$, \\$\text{depth} = \ell\log{\log{n}}$ \\(the base case for Algorithm \ref{algorithm-fast-sleeping})};
\draw [dashed] (f) -- (label2);
\node [text width=4cm, align=right] (label1) [right = of n, xshift=-2mm] {$\text{level} = i$, $\text{depth} = K - i$};
\draw [dashed] (n) -- (label1);
\node [text width=4cm, align=right] (label0) [right = of z, xshift=25mm] {$\text{level} = K = c\log{n}$, \\ $\text{depth} = 0$};
\draw [dashed] (z) -- (label0);
\node [] (source) [above left = of z, xshift=20mm, yshift=-8mm] {Execution begins at the root};

\end{tikzpicture}


\end{center}
\vspace{5 mm}
\caption{\large \boldmath The recursion tree for our algorithms (Algorithms \ref{alg:sleepingMIS} and \ref{algorithm-fast-sleeping}, respectively)}
\label{figure-recursion-tree-2}
\end{figure*}
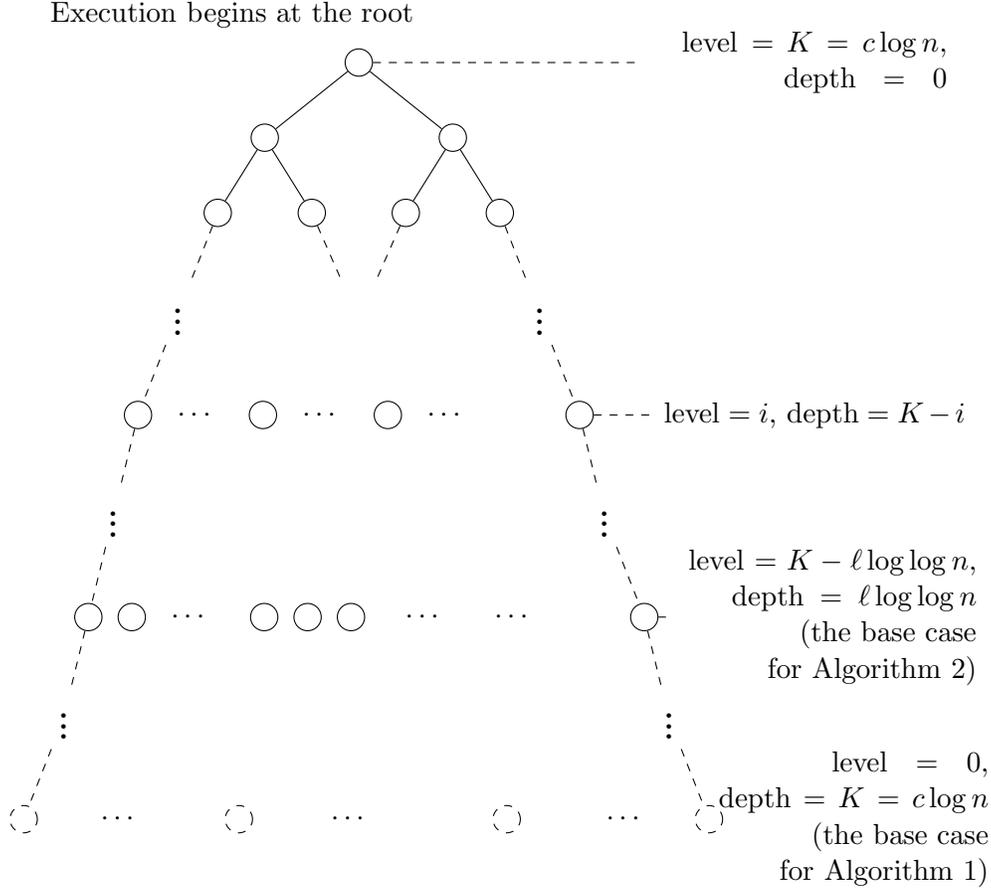}

\onlyShort{\input{figure_tree_2_SHORT}}


This means that at any given level $i$, the recursion tree corresponding to Algorithm \textsc{SleepingMISRecursive}($k$) (see Figure \ref{figure-recursion-tree-2}) would look exactly the same --- no matter which of Algorithm \textsc{SleepingMISRecursive}($k$) or the parallel/distributed randomized greedy MIS algorithm was used to solve the subsequently lower levels. In particular, Lemma \ref{lem:levelCost} would still hold if the $i^{\text{th}}$ level (for any $i  \in  [1, K]$) of the recursion tree was considered as the base case (instead of the $0^{\text{th}}$ level as in Algorithm \textsc{SleepingMISRecursive}($k$)) and the parallel/distributed randomized greedy MIS algorithm was used to solve each of the base cases. We note
that for synchronization  at higher levels of the recursion, we will require 
    that the greedy algorithm runs for (exactly) $c\log n$ rounds for some large (but fixed) constant $c > 0$. The constant $c$ is chosen such
    that the greedy algorithm finishes with high probability and the existence
    of this constant follows from the analysis of Fischer and Noever \cite{Fischer_2018}. However, this makes the algorithm Monte-Carlo,
    since there is a small probability that some base case might not finish within
    the required time which might affect the correctness of the overall algorithm.

We are now ready to present our modified algorithm which we call \textsc{Fast-SleepingMISRecursive}($k$). The recursion tree for this algorithm is shown in Figure \ref{figure-recursion-tree-2}. Note that we slightly modify the parallel/distributed randomized greedy MIS algorithm as follows: 


\onlyShort{\input{pseudocode_SHORT_second_algorithm}}

\onlyLong{\begin{algorithm}
\begin{algorithmic}[1]

\Function{SleepingMIS}{$\ell  \cdot  \log{\log{n}}$} \Comment{$\ell$ is defined as $\ell  \defeq  (\log_2{(\frac{4}{3})})^{-1}$; see Equation \ref{equation-ell}}
      \State $v.$inMIS $\gets$ \texttt{unknown}
      \For{$i$ \textbf{in} $1, \dots, \ell  \cdot  \log{\log{n}}$}
        \State $v.X_i \gets$ random bit that is $1$ with probability $\frac{1}{2}$
      \EndFor
      \State \textsc{SleepingMISRecursive}($\ell  \cdot  \log{\log{n}}$) \label{line-depth-of-recursion} \Comment{$\ell  \cdot  \log{\log{n}}$ is the recursion depth.}
\EndFunction

    \bigskip

\Function{Fast-SleepingMISRecursive}{$k$}
    \If{$k = 0$} \Comment{base case} \label{lin:baseCaseBegin}
        \State \Call{DistributedGreedyMIS}{} \Comment{Use the parallel/distributed randomized greedy MIS algorithm (see \cite[Section $2$]{Coppersmith_1989}) to solve the base cases.} 
        (For synchronization at higher levels of the recursion, we require
    that the greedy algorithm  runs for (exactly) $c\log n$ rounds for some large (but fixed) constant $c > 0$.)
        \State \textbf{return}
      \EndIf \label{lin:baseCaseEnd}

      \medskip
      \State send message to every neighbor \Comment{first isolated node detection, 1 round} \label{lin:isolatedBegin}
      \If{$v$ receives no message}
        \State $v.$inMIS $\leftarrow$ true
      \EndIf \label{lin:isolatedEnd}
       
      \medskip
      \If{$v.$inMIS $=$ unkown \textbf{and} $v.X_k = 1$} \Comment{left recursion} \label{lin:leftRecursionBegin}
        \State \textsc{SleepingMISRecursive}$(k - 1)$ \label{lin:leftRecursiveCall}
      \Else
        \State sleep for a suitable number of rounds to achieve synchronization \label{lin:leftSleep}
      \EndIf \label{lin:leftRecursionEnd}

      \medskip
      \State send value of $v.$inMIS to every neighbor  \Comment{synchronization step,  1 round} \label{lin:eliminationBegin}
      \If{$v.$inMIS $=$ unknown \textbf{and} $v$ receives message from neighbor $w$ with $w$.inMIS $=$ true}
        \State $v.$inMIS $\leftarrow$ false \Comment{elimination}
      \EndIf \label{lin:eliminationEnd}

      \medskip
      \State send value of $v.$inMIS to every neighbor \Comment{second isolated node detection, 1 round} \label{lin:isolated2Begin}
      \If{$v.$inMIS $=$ unknown \textbf{and} $v$ only receives messages from neighbors $w$ with $w.\text{inMIS} = \text{false}$}
        \State $v.$inMIS $\leftarrow$ true
      \EndIf \label{lin:isolated2End}

      \medskip
      \If{$v.$inMIS $=$ unknown} \Comment{right recursion} \label{lin:rightRecursionBegin}
        \State \textsc{SleepingMISRecursive}$(k - 1)$ \label{lin:rightRecursiveCall}
      \Else
        \State sleep for a suitable number of rounds to achieve synchronization
        \label{lin:rightSleep}
      \EndIf \label{lin:rightRecursionEnd}
    \EndFunction
\end{algorithmic}
\caption{The ``Fast Sleeping MIS algorithm'' executed by a node $v$} \label{algorithm-fast-sleeping}
\end{algorithm}}


Let's call the main algorithm that invokes \textsc{Fast-SleepingMISRecursive} as \textsc{Fast-SleepingMIS} (analogous to \textsc{SleepingMIS}) and it is now invoked with $\ell  \cdot  \log{\log{n}}$, where \onlyShort{$\ell$ is defined as in Equation \ref{equation-ell}.}

\onlyLong{
    \begin{equation} \label{equation-ell}
        \ell  \defeq  (\log_2{(\frac{4}{3})})^{-1}
    \end{equation}
}

Using Corollary \ref{cor:samemis}, the correctness of this algorithm (\textsc{Fast-SleepingMIS}) can be argued in a similar way as was done in the case of Algorithm \textsc{SleepingMIS}\onlyShort{.} \onlyLong{(see Section \ref{sec:correctness}).} \onlyShort{The full analysis of the time complexity is deferred to the longer version of the paper (in the appendix) --- for the sake of preserving space here. The main idea is to analyze the recursion tree (see Figure \ref{figure-recursion-tree-2}) to prove the various clauses of Theorem \ref{theorem-all-the-complexity-measures-for-the-second-algorithm}.}
We now proceed to analyze the time complexity of Algorithm \ref{algorithm-fast-sleeping}.

\begin{lemma} \label{lemma-algorithm-2-node-average-awake-complexity}
    The expected node-average awake complexity of \textsc{Fast-SleepingMIS}($K$) with $K = \ell  \cdot  \log{\log{n}}$ (where $\ell$ is defined as in Equation \ref{equation-ell})  in the sleeping model is $O(1)$.
\end{lemma}
\begin{proof}[Proof Sketch]
    As in the proof of Lemma \ref{lemma-first-algorithm-node-averaged-awake-complexity}, we argue that the cost of recursion is $O(1)$ per vertex, on expectation. In addition to the recursive calls, since the base cases here are non-trivial, each of them takes substantially more time, so we need to account for those too.
    
    As has been argued before, Lemma \ref{lem:levelCost} guarantees that the total (expected) number of network nodes participating at depth $\ell  \cdot  \log{\log{n}}$
    \begin{align*}
        &=   (\frac{3}{4})^{\ell  \cdot  \log{\log{n}}}  \cdot  n\\
        &=   (\frac{1}{2})^{\log{\log{n}}}  \cdot  n  \tag{substituting the value of $\ell$ from Equation \ref{equation-ell}}\\
        &=   \frac{n}{\log{n}}
    \end{align*}
    
    By virtue of \cite[Theorem $1.1$]{Fischer_2018}, the parallel/distributed randomized greedy MIS algorithm takes --- with high probability --- $O(\log{n})$ rounds to solve each instance of these base cases, i.e., recursion tree-nodes at depth $\ell  \cdot  \log{\log{n}}$. Thus each node participating at this level (depth $\ell  \cdot  \log{\log{n}}$ of the recursion tree) stays awake for $O(\log{n})$ time.
    
    Thus the total (total over all nodes) \emph{expected} awake time at the base level (i.e., at depth $\ell  \cdot  \log{\log{n}}$ of the recursion tree) is equal to
    \begin{align*}
        &=   \text{number of network nodes participating at that level}\\
        &\times   \text{awake-time of each of those participating nodes}\\
        &=   \frac{n}{\log{n}}  \times  O(\log{n})   =   O(n)\text{.}
    \end{align*}
    
    Dividing by $n$ we get the additional awake complexity per node and that turns out to be $\frac{O(n)}{n}  =  O(1)$.
    
    Thus the total \emph{expected} node-average awake complexity of \textsc{Fast-SleepingMIS}($k$) is
        \begin{align*}
            &=   \text{node-average awake complexity contributed by the recursion calls}\\
            &+  \text{node average awake complexity contributed by the solutions to the base cases}\\
            &=  O(1) + O(1)  =  O(1)\text{.}
        \end{align*}
\end{proof}

\begin{lemma} \label{lemma-algorithm-2-worst-case-round-complexity}
    It holds with high probability that the worst-case round complexity of \textsc{Fast-SleepingMIS}($K$) with
    \begin{center}
        $K = \ell  \cdot  \log{\log{n}}$
    \end{center}    
    is
    \begin{center}
        $O(\log^{\ell + 1}{(n)})  =  O(\log^{3.41}{(n)})$.
    \end{center}
\end{lemma}
\begin{proof}
    First we observe that the number of leaves of the recursion tree  =  recursion tree-nodes at depth $\ell  \cdot  \log{\log{n}}$ is $2^{\ell  \cdot  \log{\log{n}}}  =  (\log{n})^{\ell}$.
    
    Next we compute the total time taken for executing one branch of the recursion tree --- from the root to the leaf. That is given by
    \begin{align*}
        &   \text{time taken to traverse one branch of the recursion tree from the root to the leaf}\\
        &+   \text{time taken to solve the leaf, i.e., the base case of the recursion}\\
        &=   \text{depth of the recursion tree}  \times  O(1)   +   O(\log{n}) \tag{by \cite[Theorem $1.1$]{Fischer_2018}}\\
    \end{align*}
    
    The first term of the summation above equals $\ell  \cdot  \log{\log{n}}  \times  O(1)   =   O(\log{\log{n}})$ and this holds deterministically. The second term of the summation above equals $O(\log{n})$ and this holds with high probability (thanks to \cite[Theorem $1.1$]{Fischer_2018}).
    
    Substituting the values in the above equation, we get that the total time taken for executing one branch of the recursion tree --- from the root to the leaf --- is $O(\log{\log{n}})  +  O(\log{n})   =   O(\log{n})$, and this fact holds with high probability.
    
    Hence it holds with high probability that the total time taken for executing the entire recursion tree, i.e., for executing the algorithm
    \begin{align*}
        &=   \text{time taken for executing one branch of the recursion tree}   \times   \text{number of leaves of the recursion tree}\\
        &=   O(\log{n})  \times  (\log{n})^{\ell}   =   O(\log^{\ell + 1}{(n)})\\
        &=   O(\log^{3.41}{(n)})\text{.} \tag{substituting the value of $\ell$ from Equation \ref{equation-ell}}
    \end{align*}
\end{proof}

As an immediate corollary to Lemma \ref{lemma-algorithm-2-worst-case-round-complexity}, we have
\begin{lemma} \label{lemma-algorithm-2-node-averaged-round-complexity}
    It holds with high probability that the node-averaged round complexity of \textsc{Fast-SleepingMIS}($K$) with $K = \ell  \cdot  \log{\log{n}}$ is $O(\log^{\ell + 1}{(n)})  =  O(\log^{3.41}{(n)})$.
\end{lemma}

Next we argue about the worst-case awake round complexity of \textsc{Fast-SleepingMIS}($K$), with  $K = \ell  \cdot  \log{\log{n}}$. As argued  earlier, the worst-case awake round complexity is proportional to the depth of the recursion tree of the Algorithm \textsc{Fast-SleepingMISRecursive}($K$) plus the worst-case time taken by a node at the last level. The latter is $O(\log{n})$ rounds w.h.p.\ \cite[Theorem $1.1$]{Fischer_2018}, since they execute the randomized greedy algorithm with random ranks. Thus the overall worst-case awake round complexity is the depth of the recursion which is $O(\log{\log{n}})$ plus $O(\log{n})$ for the last level. Hence we can state the following.

\begin{lemma} \label{lemma-algorithm-2-worst-case-awake-complexity}
    It holds with high probability that the worst-case awake round complexity of \textsc{Fast-SleepingMIS}($K$) with $K = \ell  \cdot  \log{\log{n}}$ is $O(\log{n})$.
\end{lemma}

Combining Lemmas \ref{lemma-algorithm-2-node-average-awake-complexity}, \ref{lemma-algorithm-2-worst-case-round-complexity}, \ref{lemma-algorithm-2-worst-case-awake-complexity}, and \ref{lemma-algorithm-2-node-averaged-round-complexity} immediately gives rise to Theorem \ref{theorem-all-the-complexity-measures-for-the-second-algorithm}.

\section{Conclusion} \label{sec:conc}

In this paper, we posit the sleeping model, and show that the fundamental MIS problem on general graphs can be solved in $O(1)$ rounds under node-averaged awake complexity, i.e., the average number of awake (non-sleeping) rounds taken by all nodes is $O(1)$. This is the first such result that we are aware of where we can obtain a constant round bound for node-averaged complexity for MIS \emph{either} in the sleeping model or in the traditional model. An important open question is whether a similar bound can be shown in the traditional model; note that such a bound can be shown for $(\Delta+1)$-coloring for general graphs (see Section \ref{sec:related}). The sleeping model which is an extension of the traditional model can prove useful in designing distributed algorithms for various problems that are efficient with respect to the node-averaged measure; that would have crucial implications for saving nodes' resources (energy, etc.) in resource-constrained networks.


\section*{Acknowledgement}

We thank the anonymous reviewers at PODC 2020 for the many insightful comments.


\bibliographystyle{plain}
\bibliography{main}

\end{document}